\documentclass[journal,onecolumn,12pt]{IEEEtran}
\usepackage{pslatex}
\usepackage{amsfonts,color,morefloats}
\usepackage{amssymb,amsmath,latexsym,amsthm}
\usepackage{lineno}

\newcommand{\wt}{{\mathrm{wt}}}

\newcommand{\tr}{{\mathrm{Tr}}}
\newcommand{\gf}{{\mathrm{GF}}}

\newcommand{\C}{{\mathcal{C}}}

\newcommand{\bc}{{\mathbf{c}}}

\newtheorem{theorem}{Theorem}
\newtheorem{lemma}[theorem]{Lemma}

\newtheorem{corollary}[theorem]{Corollary}

\newtheorem{example}{Example}
\newtheorem{remark}{Remark}

\begin{document}

\title{Linear codes with a few weights from inhomogeneous quadratic functions \thanks{The research of C. Tang was supported by China West Normal University under Grant 14E013 and Grant CXTD2014-4. The research of K. Feng was supported by NSFC No. 11471178, 11571007 and the Tsinghua National Lab. for Information Science and Technology.}}

\author{
Chunming Tang\thanks{C. Tang is with School of Mathematics and Information, China West Normal University, Sichuan Nanchong, 637002, China. e-mail: tangchunmingmath@163.com},
Can Xiang\thanks{C. Xiang is with the College of Mathematics and Information Science, Guangzhou University, Guangzhou 510006, China. Email: cxiangcxiang@hotmail.com} and
Keqin Feng\thanks{K. Feng is with the Department of Mathematical Sciences, Tsinghua University, Beijing, 100084, China.
Email: kfeng@math.tsinghua.edu.cn.}
}

\date{\today}
\maketitle

\begin{abstract}
Linear codes with few weights have been an interesting subject of study for many years, as these codes have applications in secrete sharing, authentication codes, association schemes, and strongly regular graphs. In this paper, linear codes with a few weights are constructed from inhomogeneous quadratic functions over the finite field $\gf(p)$, where $p$ is an odd prime. They include some earlier linear codes as
special cases. The weight distributions of these linear codes are also determined.
\end{abstract}


\section{Introduction}\label{sec-intro}

Throughout this paper, let $p$ be an odd prime and let $q=p^m$ for some positive integer $m$.
An $[n,\, k,\,d]$ code $\C$ over $\gf(p)$ is a $k$-dimensional subspace of $\gf(p)^n$ with minimum
(Hamming) distance $d$.  Let $A_i$ denote the number of codewords with Hamming weight $i$ in a code
$\C$ of length $n$. The {\em weight enumerator} of $\C$ is defined by
$
1+A_1z+A_2z^2+ \cdots + A_nz^n.
$
The {\em weight distribution} $(1,A_1,\ldots,A_n)$ is an important research topic in coding theory,
as it contains crucial information as to estimate the error correcting capability and the probability of
error detection and correction with respect to some algorithms.
A code $\C$ is said to be a $t$-weight code  if the number of nonzero
$A_i$ in the sequence $(A_1, A_2, \cdots, A_n)$ is equal to $t$.

Let $\tr$ denote the trace function from $\gf(q)$ onto $\gf(p)$ throughout
this paper. Let $F(x)\in \gf(q)[x]$,
$D=\{x\in \gf(q)^*: \tr(F(x))=0\}=\{d_1, \,d_2, \,\ldots, \,d_n\} \subseteq \gf(q)$ and $n=\#D$.
We define a linear code of
length $n$ over $\gf(p)$ by
\begin{eqnarray}\label{eqn-maincode}
\C_{D}=\{(\tr(xd_1), \tr(xd_2), \ldots, \tr(xd_n)): x \in \gf(q)\},
\end{eqnarray}
and call $D$ the \emph{defining set} of this code $\C_{D}$. By definition, the
dimension of the code $\C_D$ is at most $m$.

This construction is generic in the sense that many classes of known codes
could be produced by properly selecting the defining set $D \subseteq \gf(q)$. If
the defining set $D$ is well chosen, some optimal linear codes with few weights can be obtained. Based on this construction, many linear codes have been constructed since Ding et al. published their paper in 2014 \cite{DingDing1}. We refer interested readers to \cite{Ding15,DingDing2,Ding20152,Mesnager2015,ZLFH2015,WDX2015,TLQZH2015,TQH2015,QTH2015} and the references therein.
Particularly,
Ding et al. \cite{DingDing2} presented the
weight distribution of
$\mathcal{C}_D$ for the case $F(x)=
x^2$ and proposed an open problem
on how to determine the weight distribution of  $\mathcal{C}_{D}$ for general planar functions
$F(x)$. Subsequently, Zhou et al. \cite{ZLFH2015} and Tang et al. \cite{TLQZH2015} solved this open problem and gave the weight distribution of
$\mathcal{C}_D$ from homogeneous quadratic Bent functions and weakly regular Bent functions with some homogeneous conditions, respectively.

In this paper, we consider linear codes
with few weights from inhomogeneous quadratic functions $\tr(F(x))=f(x)-\tr(\alpha x)$ and determine the weight distributions of these linear codes , where
$\alpha \in \gf(q)$, $f(x)$ is a homogeneous quadratic function from $\gf(q)$ onto $\gf(p)$ and defined by

\begin{eqnarray}\label{def-f}
f(x)=\sum_{i=0}^{m-1}\tr(a_i x^{p^i+1})~~~(a_i\in \gf(q)).
\end{eqnarray}
They include some earlier linear codes as special cases \cite{XTF2015,LWL2015}.

The rest of this paper is organized as follows. Section \ref{sec-pr} introduces some basic notations and results of group characters, Gauss
sums, exponential sums and cyclotomic fields which will be needed in subsequent sections. Section \ref{sec-main} constructs linear codes with a few weights from inhomogeneous quadratic functions and settles the
weight distributions of these linear codes. Section \ref{sec-concluding} summarizes this paper.

\section {Preliminaries}\label{sec-pr}

In this section, we state some notations and basic facts on group characters, Gauss
sums, exponential sums and cyclotomic fields. Moreover, we give and prove some results on exponential sums about homogeneous quadratic functions $f(x)$ defined in (\ref{def-f}). These results will be used in the rest of the paper.

\subsection{Some notations fixed throughout this paper}

For convenience, we adopt the following notations unless otherwise stated in this paper.
\begin{itemize}
\item $p^*=(-1)^{(p-1)/2}p$.
  \item $\zeta_p=e^{\frac{2\pi \sqrt{-1}}{p}}$ be the primitive $p$-th
root of unity.

\item $\textup{SQ}$ and $\textup{N\textup{SQ}}$ denote the set of all  squares and nonsquares in $\gf(p)^{*}$, respectively.
\item $\eta$ and $\bar{\eta}$ are the quadratic characters of $\gf(q)^{*}$ and  $\gf(p)^{*}$, repsectively. We extend these quadratic characters
by letting $\eta(0)=0$ and $\bar{\eta}(0)=0$.
\end{itemize}

\subsection{Group characters and Gauss sums}

An {\em additive character} of $\gf(q)$ is a nonzero function $\chi$
from $\gf(q)$ to the set of nonzero complex numbers such that
$\chi(x+y)=\chi(x) \chi(y)$ for any pair $(x, y) \in \gf(q)^2$.
For each $b\in \gf(q)$, the function
\begin{eqnarray}\label{dfn-add}
\chi_b(c)=\zeta_p^{\tr(bc)} \ \ \mbox{ for all }
c\in\gf(q)
\end{eqnarray}
defines an additive character of $\gf(q)$. When $b=0$,
$\chi_0(c)=1 \mbox{ for all } c\in\gf(q),
$
and is called the {\em trivial additive character} of
$\gf(q)$. The character $\chi_1$ in (\ref{dfn-add}) is called the
{\em canonical additive character} of $\gf(q)$.
It is well known that every additive character of $\gf(q)$ can be
written as $\chi_b(x)=\chi_1(bx)$ \cite[Theorem 5.7]{LN}.

The Gauss sum $G(\eta, \chi_1)$ over $\gf(q)$ is defined by
\begin{eqnarray}
G(\eta, \chi_1)=\sum_{c \in \gf(q)^*} \eta(c) \chi_1(c) = \sum_{c \in \gf(q)} \eta(c) \chi_1(c)
\end{eqnarray}
and
the Gauss sum $G(\bar{\eta}, \bar{\chi}_1)$ over $\gf(p)$ is defined by
\begin{eqnarray}
G(\bar{\eta}, \bar{\chi}_1)=\sum_{c \in \gf(p)^*} \bar{\eta}(c) \bar{\chi}_1(c)
= \sum_{c \in \gf(p)} \bar{\eta}(c) \bar{\chi}_1(c),
\end{eqnarray}
where $\bar{\chi}_1$ is the canonical additive characters of $\gf(p)$.

The following three lemmas are proved in \cite[Theorem 5.15 and Theorem 5.33]{LN} and \cite[lemma 7]{DingDing2}, respectively.

\begin{lemma}\label{lem-32A1}
With the symbols and notations above, we have
$$
G(\eta, \chi_1)=(-1)^{m-1} \sqrt{-1}^{(\frac{p-1}{2})^2 m} \sqrt{q}
$$
and
$$
G(\bar{\eta}, \bar{\chi}_1)= \sqrt{-1}^{(\frac{p-1}{2})^2 } \sqrt{p}=\sqrt{p*}.
$$
\end{lemma}

\begin{lemma}\label{lem-32A2}
Let $\chi$ be a nontrivial additive character of $\gf(q)$ with $q$ odd, and let
$f(x)=a_2x^2+a_1x+a_0 \in \gf(q)[x]$ with $a_2 \ne 0$. Then
$$
\sum_{c \in \gf(q)} \chi(f(c)) = \chi(a_0-a_1^2(4a_2)^{-1}) \eta(a_2) G(\eta, \chi).
$$
\end{lemma}


\subsection{Cyclotomic fields}
In this subsection, we state some basic facts on cyclotomic fields. These results will be used in the rest of this paper.

Let $\mathbb{Z}$ be the rational integer ring and $Q$ be the rational field.
Some results on
cyclotomic field $Q(\zeta_p)$ \cite{IR1990} are given in the following lemma.
\begin{lemma}\label{cyclo} We have the following basic facts.
\begin{enumerate}
 \item
The ring of integers
in $K=Q(\zeta_p)$ is $\mathcal{O}_K=
\mathbb{Z}(\zeta_p)$ and $\{
\zeta_p^{~i}: 1\leq i\leq p-1\}$
is an integral basis of $\mathcal{O}_K$.
\item
The field extension $K/Q$
is Galois of degree $p-1$ and the Galois
group $Gal(K/Q)=\{\sigma_a:
a\in (\mathbb{Z}/p\mathbb{Z})^{*}\}$, where
the automorphism $\sigma_a$ of $K$ is defined by
$\sigma_a(\zeta_p)=\zeta_p^a$.
\item
The field $K$ has a unique
quadratic subfield $L=Q(\sqrt{p^*})$. For $1\leq a\leq p-1$,
$\sigma_a(\sqrt{p^*}) =
\bar{\eta}(a)\sqrt{p^*}$. Therefore, the Galois group
$Gal(L/Q)$ is $\{1,\sigma_{\gamma}\}$, where
$\gamma$ is any quadratic nonresidue in
$\gf(p)$.
\end{enumerate}
\end{lemma}

From Lemma \ref{cyclo}, the conclusion of the following lemma is straightforward and we omit their proofs.
\begin{lemma}\label{lem-A2}With the symbols and notations above, we have the following.

(\uppercase\expandafter{\romannumeral1})
$
\sum_{y \in \gf(p)^*}\sigma_y((p^*)^{-\frac{r}{2}})
=
\left\{ \begin{array}{ll}
0     & \mbox{ if $r$ is odd,} \\
(p^*)^{-\frac{r}{2}}(p-1)     & \mbox{ if $r$ is even.}
\end{array}
\right.
$

(\uppercase\expandafter{\romannumeral2}) For any $z\in\gf(p)^*$, then
\begin{eqnarray*}
\sum_{y \in \gf(p)^*}\sigma_y((p^*)^{-\frac{r}{2}}~\zeta_p^{~z})
&=&\left\{ \begin{array}{ll}
\bar{\eta}(z)(p^*)^{-\frac{r-1}{2}}    & \mbox{ if $r$ is odd,}\\
-(p^*)^{-\frac{r}{2}}     & \mbox{ if $r$ is even.}
\end{array}
\right.
\end{eqnarray*}

\end{lemma}

\subsection{Exponential sums}

In this subsection, we give and prove some results on exponential sums about homogeneous quadratic functions $f(x)$ defined in (\ref{def-f}). Before doing this, we need state some basic facts on linear algebra.

The field $\gf(q)$ is a vector space over $\gf(p)$ with dimension $m$. We fix a basis $v_0, v_1, . . . , v_{m-1}$ of $\gf(q)$
over $\gf(p)$. Then each $x\in \gf(q)$ can be uniquely expressed as
$$
x=x_0v_0+x_1v_1+\cdots+x_{m-1}v_{m-1}~~~(x_i\in\gf(p)).
$$
Thus we have the following $\gf(p)$-linear isomorphism $\gf(q)\xrightarrow{\thicksim}\gf(p)^m$:
$$
x=x_0v_0+x_1v_1+\cdots+x_{m-1}v_{m-1} \mapsto X=(x_0,x_1,\cdots,x_{m-1}).
$$
With this isomorphism, a function $f:\gf(q) \rightarrow \gf(p)$ induces a function $F:\gf(p)^m \rightarrow \gf(p)$ where for
all $X=(x_0,x_1,\cdots,x_{m-1})\in \gf(p)^m$, $F(X) = f (x)$ where $x=x_0v_0+x_1v_1+\cdots+x_{m-1}v_{m-1}$. In this way, the function $f$ defined in (\ref{def-f}) induces a quadratic form
\begin{eqnarray}\label{eqn-H}
F(X) &=& \sum_{i=0}^{m-1}\tr(a_i(\sum_{j=0}^{m-1}x_jv_j)^{p^i+1}) \nonumber \\
&=& \sum_{i=0}^{m-1}\tr(a_i(\sum_{j=0}^{m-1}x_jv_j^{p^i})(\sum_{k=0}^{m-1}x_kv_k))  \nonumber \\
&=&\sum_{j=0}^{m-1}\sum_{k=0}^{m-1}(\sum_{i=0}^{m-1}\tr(a_iv_j^{p^i}v_k))x_jx_k  \nonumber \\
&=& XHX^{T},
\end{eqnarray}
where $X^{T}$ is the transposition of $X$, $H=(h_{j,k})$,
$$
h_{j,k}=\frac{1}{2}\sum_{i=0}^{m-1}(\tr(a_i(v_j^{p^i}v_k+v_jv_k^{p^i}))) ~~\rm{~for~} 0\leq j,k\leq m-1,
$$
and the rank of $H$ is called the rank of the function $f$ defined in (\ref{def-f}). We denote the rank of $f$ by $r_f$. Particularly, $r_f=m$ if and only if $f$ is Bent function.

Since $H$ defined in (\ref{eqn-H}) is a $m \times m$ symmetric matrix over $\gf(p)$ and $r_f = \textmd{rank}~H$, there exists $M\in \textmd{GL}_m(\gf(p))$ such that $H'=MHM^{T}$ is a diagonal matrix and $H'= diag(\lambda_1,\cdots, \lambda_{r_f}, 0, \cdots, 0)$ where $\lambda_i\in \gf(p)^*$($1\leq i\leq r_f$). Let $\Delta=\lambda_1,\cdots,\lambda_{r_f}$. Then the value of $\bar{\eta}(\Delta)$ is an invariant of $H$ under the action of $H\mapsto MHM^{T}$ where $M\in \textmd{GL}_m(\gf(p))$. We call $\bar{\eta}(\Delta)$ the sign of the quadratic function $f$ of (\ref{def-f}) and is defined by $\varepsilon_f$.

It is clear that the value of $r_f$  is closely related to the value of $\#Z_f$, where the set
$$
Z_f~=\{x\in \gf(q):f(x+y)=f(x)+f(y),\forall y\in \gf(q)\}.
$$
It is well known that $\#Z_f=p^{m-r_f}$. Note that from Equation (\ref{def-f}) we have
\begin{eqnarray}\label{eqn-flf}
f(x+y)=f(x)+f(y)+2\tr(L_f(x)y)=f(x)+f(y)+2\tr(xL_f(y)),
\end{eqnarray}
where $L_f$ is a linear polynominal over $\gf(q)$ defined by
$$
L_f(x)=\frac{1}{2}\sum_{i=0}^{m-1}(a_i+a_{m-i}^{p^i})x^{p_i}.
$$

From now on we define $\textup{Im}(L_f)=\{L_f(x):x\in \gf(q)\}$ and $\textup{Ker}(L_f)=\{x\in \gf(q): L_f(x)=0\}.$ If $b\in \textup{Im}(L_f)$, we denote $x_b\in \gf(q)$ with satisfying $L_f(x_b)=-\frac{b}{2}$.

From Equation (\ref{eqn-flf}), we have
$$
\mbox{ker}(L_f)=\{x\in \gf(q): f(x+y)=f(x)+f(y)~~\mbox{for all~~} y\in \gf(q)\}.
$$
Thus $p^{m-r_f}=\# Z_f=\#\textmd{Ker}(L_f)$, that is, $\textmd{rank~} L_f=r_f$. It is obvious that $0\leq r_f\leq m$.

\section{Linear codes from inhomogeneous quadratic functions}\label{sec-main}

We construct linear codes over $\gf(p)$ by using inhomogeneous quadratic functions  and determine their parameters in this section.

In this paper, the defining set $D$ of the code $\C_D$ of (\ref{eqn-maincode})  is given by
\begin{eqnarray}\label{eqn-defsetD}
D=\{x \in \gf(q)^*: f(x)-\tr(\alpha x)=0\},
\end{eqnarray}
where $\alpha \in \gf(q)^*$ and $f$ is defined in (\ref{def-f}). It is clear that the function $f(x)-\tr(\alpha x)$ used in the defining set $D$ is a inhomogeneous quadratic functions.

Before giving and proving the main results of this paper, we firstly prove a few more auxiliary results which will be needed in proving the main results.

\subsection{Some auxiliary results}

To prove our main results in this paper, we need the help of a number of lemmas that are described and proved
in this subsection.

\begin{lemma}\label{lem-esf}

Let the symbols and notations be as above. Let $f$ be a homogeneous quadratic function and $b\in \gf(q)$. Then

(\uppercase\expandafter{\romannumeral1}) $\sum_{x\in \gf(q)}\zeta_p^{f(x)}=\varepsilon_f p^m(p^*)^{-\frac{r_f}{2}}$ and

(\uppercase\expandafter{\romannumeral2}) $\sum_{x\in \gf(q)}\zeta_p^{f(x)-\tr(bx)}
=\left\{ \begin{array}{ll}
0     & \mbox{ if $b \not\in \textmd{Im}(L_f)$} \\
\varepsilon_f p^m(p^*)^{-\frac{r_f}{2}}\zeta_p^{-f(x_b)} & \mbox{ if $b \in \textmd{Im}(L_f)$}
\end{array}
\right.$,
where $x_b$ satisfies $L_f(x_b)=-\frac{b}{2}$.
\end{lemma}

\begin{proof}
(\uppercase\expandafter{\romannumeral1}) The desired conclusion (\uppercase\expandafter{\romannumeral1}) of this lemma then follows from \cite[Lemma 1]{FL07}.

(\uppercase\expandafter{\romannumeral2}) If $b \not\in \textmd{Im}(L_f)$, then we have
\begin{eqnarray}\label{eqn-bnotin}
& &\left(\sum_{x\in \gf(q)}\zeta_p^{-f(x)}\right)\left(\sum_{y\in \gf(q)}\zeta_p^{f(y)-\tr(by)} \right) \nonumber \\
& & =\sum_{x\in \gf(q)}\zeta_p^{-f(x)} \sum_{y\in \gf(q)}\zeta_p^{f(x+y)-\tr(b(x+y))} \nonumber \\
& & = \sum_{x,y\in \gf(q)}\zeta_p^{f(x+y)-f(x)-\tr(b(x+y))}  \nonumber \\
& & =\sum_{x,y\in \gf(q)}\zeta_p^{f(y)+2\tr(L_f(y)x)-\tr(b(x+y))}  ~~~~~~~~~~(\mbox{By Equation (\ref{eqn-flf})}) \nonumber \\
& & = \sum_{y\in \gf(q)}\zeta_p^{f(y)-\tr(by)}\sum_{x\in \gf(q)}\zeta_p^{\tr((L_f(2y)-b)x)} \nonumber \\
& & = 0.                      ~~~~~~~~~~~~~~~~~~~~~~~~~~~~~~~~~~~~~~~~~~~~~~~(\mbox{Since $b \not\in \textmd{Im}(L_f)$}) \nonumber
\end{eqnarray}
From the conclusion (\uppercase\expandafter{\romannumeral1}) of this lemma, we have $\sum_{x\in \gf(q)}\zeta_p^{-f(x)}\neq 0$. Therefore, $\sum_{y\in \gf(q)}\zeta_p^{f(y)-\tr(by)}=0$.

If $b \in \textmd{Im}(L_f)$, then there exists $x_b\in \gf(q)$ such that $L_f(x_b)=-\frac{b}{2}$. Thus, we have

\begin{eqnarray}\label{eqn-bin}
 \sum_{x\in \gf(q)}\zeta_p^{f(x)-\tr(bx)}
&=& \sum_{x\in \gf(q)}\zeta_p^{f(x)+2\tr(L_f(x_b)x)}  \nonumber \\
&=& \sum_{x\in \gf(q)}\zeta_p^{f(x)+f(x_b)+2\tr(L_f(x_b)x)-f(x_b)}  \nonumber \\
&=& \zeta_p^{-f(x_b)}\sum_{x\in \gf(q)}\zeta_p^{f(x+x_b)} ~~~~~~~~~~(\mbox{By Equation (\ref{eqn-flf})}) \nonumber \\
&=& \zeta_p^{-f(x_b)}\sum_{x\in \gf(q)}\zeta_p^{f(x)}        \nonumber \\
&=& \varepsilon_f p^m(p^*)^{-\frac{r_f}{2}}\zeta_p^{-f(x_b)}.  ~~~~~~~~~~~(\mbox{By the conclusion (\uppercase\expandafter{\romannumeral1}) of this lemma}) \nonumber
\end{eqnarray}

Summarizing all the conclusions above, this completes the proof of this lemma.
\end{proof}

\begin{lemma}\label{lem-szw}

Let $a,b,c\in \gf(p)$ and
$$
S=\sum_{z,w\in \gf(p)}\zeta_p^{az^2+2bzw+cw^2}.
$$
Then we have the following.

(\uppercase\expandafter{\romannumeral1}) If $ac-b^2\neq 0$, then $S=\bar{\eta}(ac-b^2)p^2(p^*)^{-1}$.

(\uppercase\expandafter{\romannumeral2}) If $ac-b^2= 0$ and $a\neq0$, then $S=\bar{\eta}(a)p\sqrt{p^*}$.
\end{lemma}

\begin{proof}
(\uppercase\expandafter{\romannumeral1}) The desired conclusion (\uppercase\expandafter{\romannumeral1}) of this lemma then follows from \cite[Lemma 1]{FL07}.

(\uppercase\expandafter{\romannumeral2}) If $ac-b^2= 0$ and $a\neq0$, then
\begin{eqnarray*}
S
&=& \sum_{z,w\in \gf(p)}\zeta_p^{\frac{1}{a}(az+bw)^2}  \\
&=& \sum_{w\in \gf(p)}\sum_{z\in \gf(p)}\zeta_p^{\frac{1}{a}(az+bw)^2}  \\
&=& \sum_{w\in \gf(p)}\sum_{z\in \gf(p)}\zeta_p^{z^2}  \\
&=& \bar{\eta}(a)p\sqrt{p^*},
\end{eqnarray*}
where the last identity follows from Lemmas \ref{lem-32A1} and \ref{lem-32A2}.

This completes the proof of this lemma.
\end{proof}

\begin{lemma}\label{lem-gt}
Let $g$ be a homogeneous quadratic function from $\gf(q)$ onto $\gf(p)$ with the rank $r_g$ and the sign $\varepsilon_g$. For any $t\in \gf(p)^*$, let
$$
N(g=t)=\#\{x\in \gf(q):g(x)=t \}.
$$
Then
\begin{eqnarray*}
N(g=t)=
\left\{ \begin{array}{ll}
p^{m-1}-\varepsilon_g p^{m-1}(p^*)^{-\frac{r_g}{2}}                 & \mbox{ if $r_g$ is even},\\
p^{m-1}+\varepsilon_g \bar{\eta}(-t)p^{m-1}(p^*)^{-\frac{r_f-1}{2}} & \mbox{ if $r_g$ is odd}.
\end{array}
\right.
\end{eqnarray*}
\end{lemma}

\begin{proof}
By definition, we have
\begin{eqnarray*}
N(g=t)
&=& p^{-1} \sum_{x\in \gf(q)}\sum_{y\in \gf(p)}\zeta_p^{y(g(x)-t)}   \\
&=& p^{-1}\left(\sum_{x\in \gf(q)}\zeta_p^{0}+\sum_{y\in \gf(p)^*}\sigma_y\ (\sum_{x\in \gf(q)}\zeta_p^{g(x)-t})\right)\\
&=& p^{m-1}+p^{-1}\sum_{y\in \gf(p)^*}\sigma_y\ (\zeta_p^{-t} \varepsilon_g p^{m}(p^*)^{-\frac{r_g}{2}})~~~~~~~~~~~~~~~~(\mbox{By Lemma \ref{lem-esf}})    \\
&=& p^{m-1}+\varepsilon_g p^{m-1}\sum_{y\in \gf(p)^*}\sigma_y\ (\zeta_p^{-t}(p^*)^{-\frac{r_g}{2}}).
\end{eqnarray*}
The desired conclusion then follows from the result (\uppercase\expandafter{\romannumeral2}) of Lamma \ref{lem-A2}.
\end{proof}

\begin{lemma}\label{lem-fato}
Let $f$ be a homogeneous quadratic function with the rank $r_f$ and the sign $\varepsilon_f$, $\alpha \in Im(L_f)$ and $x_\alpha \in \gf(q)$ with satisfying $L_f(x_\alpha)=-\frac{\alpha}{2}$. Let $f(x_\alpha)=0$ and
$$
A=\#\{x\in \gf(q): f(x)=a ~and~ \tr(\alpha x)=0\}
$$
for any $a\in \gf(p)^*$. Then
$$
A=p^{m-2}+\varepsilon_f \bar{\eta}(-a)p^{m-1}(p^*)^{-\frac{r_f-1}{2}}.
$$
\end{lemma}

\begin{proof}
By definition, we have
\begin{eqnarray*}
A &=& p^{-2} \sum_{x\in \gf(q)}(\sum_{y\in \gf(p)}\zeta_p^{y(f(x)-a)})(\sum_{z\in \gf(p)}\zeta_p^{z\tr(\alpha x)}) \\
&=& p^{-2} \sum_{x\in \gf(q)}(\sum_{z\in \gf(p)}\zeta_p^{z\tr(\alpha x)})+p^{-2} \sum_{y\in \gf(p)^*}(\sum_{z\in \gf(p)}\sum_{x\in \gf(q)}\zeta_p^{y(f(x)-a)+z\tr(\alpha x)})\\
&=& p^{m-2} + p^{-2} \sum_{y\in \gf(p)^*}\sigma_y(\sum_{z\in \gf(p)}\zeta_p^{-a}\sum_{x\in \gf(q)}\zeta_p^{f(x)+z\tr(\alpha x)})\\
&=& p^{m-2} + p^{-2} \sum_{y\in \gf(p)^*}\sigma_y(\zeta_p^{-a}\sum_{z\in \gf(p)} \zeta_p^{-f(x_\alpha)z^2}\varepsilon_fp^m(p^*)^{-\frac{r_f}{2}}) ~~~~~~~~~~~~~~~(\mbox{By Lemma \ref{lem-esf}})\\
&=& p^{m-2} + p^{-2} \sum_{y\in \gf(p)^*}\sigma_y(\zeta_p^{-a}\varepsilon_fp^{m+1}(p^*)^{-\frac{r_f}{2}}) ~~~~~~~~~~~~~~~~~~~~~~~~~~~~~~~(\mbox{Since $f(x_\alpha)=0$})\\
&=& p^{m-2} + p^{-2} \bar{\eta}(-a)\varepsilon_fp^{m+1}(p^*)^{-\frac{r_f-1}{2}}) \\
&=& p^{m-2}+\varepsilon_f \bar{\eta}(-a)p^{m-1}(p^*)^{-\frac{r_f-1}{2}}.
\end{eqnarray*}

This completes the proof.
\end{proof}

\begin{lemma}\label{lem-Nf1}

Let the symbols and notations be as above. Let $f$ be a homogeneous quadratic function, $\alpha \in \gf(q)$ and
$$
N_f(\alpha)=\#\{x\in \gf(q):f(x)-\tr(\alpha x)=0\}.
$$
Then we have the following.

(\uppercase\expandafter{\romannumeral1}) If $\alpha \not\in \textmd{Im}(L_f)$, then $N_f(\alpha)=p^{m-1}$.

(\uppercase\expandafter{\romannumeral2}) If $\alpha \in \textmd{Im}(L_f)$, then
$$
N_f(\alpha)=
\left\{ \begin{array}{ll}
p^{m-1}+\varepsilon_f (p-1)p^{m-1}(p^*)^{-\frac{r_f}{2}}    & \mbox{ if $r_f$ is even and $f(x_\alpha)=0$,} \\
p^{m-1}-\varepsilon_f p^{m-1}(p^*)^{-\frac{r_f}{2}}    & \mbox{ if $r_f$ is even and $f(x_\alpha)\neq 0$,} \\
p^{m-1}                                                & \mbox{ if $r_f$ is odd and $f(x_\alpha)=0$,} \\
p^{m-1}+1+\varepsilon_f\bar{\eta}(-f(x_\alpha)) p^{m-1}(p^*)^{-\frac{r_f-1}{2}}    & \mbox{ if $r_f$ is odd and $f(x_\alpha)\neq 0$,}
\end{array}
\right.
$$
where $x_\alpha$ satisfies $L_f(x_\alpha)=-\frac{\alpha}{2}$, $r_f$ is the rank of $f$ and $\varepsilon_f$ is the sign of $f$.
\end{lemma}

\begin{proof}
By definition, we have
\begin{eqnarray*}
N_f(\alpha)
&=& p^{-1} \sum_{x\in \gf(q)}\sum_{y\in \gf(p)}\zeta_p^{y(f(x)-\tr(\alpha x))}  \\
&=& p^{m-1}+p^{-1}\sum_{y\in \gf(p)^*}\sigma_y\left( \sum_{x\in \gf(q)}\zeta_p^{f(x)-\tr(\alpha x)} \right).
\end{eqnarray*}
The desired conclusions then follow from Lemma \ref{lem-A2} and the result (\uppercase\expandafter{\romannumeral2}) of Lemma \ref{lem-esf}.
\end{proof}

\begin{lemma}\label{lem-ftr} Let the symbols and notations be as above. Let $f$ be a homogeneous quadratic function with the rank $r_f$ and the sign $\varepsilon_f$, $\beta \in \gf(q)^*$ and
\begin{eqnarray*}
S_1 &=&\sum_{z\in \gf(p)}\sum_{x\in \gf(q)}\zeta_p^{-z\tr(\beta x)}, \\
S_2 &=&\sum_{z\in \gf(p)}\sum_{x\in \gf(q)}\zeta_p^{f(x)-z\tr(\beta x)}, \\
S_3 &=& \sum_{y\in \gf(p)^*}\sigma_y\left(\sum_{z\in \gf(p)}\sum_{x\in \gf(q)}\zeta_p^{f(x)-z\tr(\beta x)}\right).
\end{eqnarray*}
Then we have the following:

(\uppercase\expandafter{\romannumeral1}) $S_1=q$,

(\uppercase\expandafter{\romannumeral2})
$S_2
=\left\{ \begin{array}{ll}
\varepsilon_f p^{m+1}(p^*)^{-\frac{r_f}{2}}   & \mbox{ if $\beta \in \textmd{Im}(L_f)$ and $f(x_\beta)=0$} \\
\varepsilon_f \bar{\eta}(-f(x_\beta))p^{m}(p^*)^{-\frac{r_f-1}{2}}   & \mbox{ if $\beta \in \textmd{Im}(L_f)$ and $f(x_\beta)\neq0$} \\
\varepsilon_f p^m(p^*)^{-\frac{r_f}{2}}                              & \mbox{ if $\beta \not \in \textmd{Im}(L_f)$}
\end{array}
\right.$,

(\uppercase\expandafter{\romannumeral3}) if $r_f$ is even, then
$$
S_3
=\left\{ \begin{array}{ll}
\varepsilon_f (p-1)p^{m+1}(p^*)^{-\frac{r_f}{2}}   & \mbox{ if $\beta \in \textmd{Im}(L_f)$ and $f(x_\beta)=0$}, \\
0                                                  & \mbox{ if $\beta \in \textmd{Im}(L_f)$ and $f(x_\beta)\neq0$}, \\
\varepsilon_f (p-1)p^m(p^*)^{-\frac{r_f}{2}}                              & \mbox{ if $\beta \not \in \textmd{Im}(L_f)$,}
\end{array}
\right.
$$
if  $r_f$ is odd, then
$$
S_3
=\left\{ \begin{array}{ll}
0                                                  & \mbox{ if $\beta \in \textmd{Im}(L_f)$ and $f(x_\beta)=0$, or $\beta \not \in \textmd{Im}(L_f)$ }, \\
\varepsilon_f \bar{\eta}(-f(x_\beta))(p-1)p^{m}(p^*)^{-\frac{r_f-1}{2}}   & \mbox{ if $\beta \in \textmd{Im}(L_f)$ and $f(x_\beta)\neq0$}.
\end{array}
\right.
$$
where $x_\beta \in \gf(q)$ satisfies $L_f(x_\beta)=-\frac{\beta}{2}$ when $\beta \in \textmd{Im}(L_f)$.
\end{lemma}

\begin{proof}(\uppercase\expandafter{\romannumeral1}) Note that
$$
\sum_{z\in \gf(p)^*}\sum_{x\in \gf(q)}\zeta_p^{\tr(-z\beta x)}=0,
$$
as $\beta \in \gf(q)^*$.
Therefore, we have
\begin{eqnarray*}
S_1
&=&\sum_{x\in \gf(q)}\zeta_p^{0}+\sum_{z\in \gf(p)^*}\sum_{x\in \gf(q)}\zeta_p^{-z\tr(\beta x)}\\
&=&q.
\end{eqnarray*}

(\uppercase\expandafter{\romannumeral2}) By definitions and the result (\uppercase\expandafter{\romannumeral2}) of Lemma \ref{lem-esf}, we have
$$
S_2
=\left\{ \begin{array}{ll}
\sum_{z\in \gf(p)} \varepsilon_f  p^{m}(p^*)^{-\frac{r_f}{2}} \zeta_p^{-f(x_\beta)z^2}    & \mbox{ if $\beta \in \textmd{Im}(L_f)$}, \\
\sum_{x\in \gf(q)}\zeta_p^{f(x)}   & \mbox{ if $\beta \not \in \textmd{Im}(L_f)$}.
\end{array}
\right.
$$
The desired conclusion (\uppercase\expandafter{\romannumeral2}) of this lemma then follows from Lammas \ref{lem-32A1} and \ref{lem-32A2} and the result (\uppercase\expandafter{\romannumeral1}) of Lemma \ref{lem-esf}.

(\uppercase\expandafter{\romannumeral3}) The desired conclusion then follows directly from Lamma \ref{lem-A2} and the result (\uppercase\expandafter{\romannumeral2}) of this lemma.

This completes the proof.
\end{proof}

\begin{lemma}\label{lem-Nf2}

Let the symbols and notations be as above. Let $f$ be a homogeneous quadratic function with the rank $r_f$ and the sign $\varepsilon_f$, $\beta \in \gf(q)^*$ and
$$
N_{f,\beta}=\#\{x\in \gf(q):f(x)=0 ~and~ \tr(\beta x)=0\}.
$$
Then, for the case $r_f$ being even, we have
$$
N_{f,\beta}
=\left\{ \begin{array}{ll}
p^{m-2}+\varepsilon_f (p-1)p^{m-1}(p^*)^{-\frac{r_f}{2}}   & \mbox{ if $\beta \in \textmd{Im}(L_f)$ and $f(x_\beta)=0$}, \\
p^{m-2}                                                 & \mbox{ if $\beta \in \textmd{Im}(L_f)$ and $f(x_\beta)\neq0$}, \\
p^{m-2}+\varepsilon_f (p-1)p^{m-2}(p^*)^{-\frac{r_f}{2}}                              & \mbox{ if $\beta \not \in \textmd{Im}(L_f)$,}
\end{array}
\right.
$$
and for the case $r_f$ being odd, we have
$$
N_{f,\beta}
=\left\{ \begin{array}{ll}
p^{m-2}                  & \mbox{ if $\beta \in \textmd{Im}(L_f)$ and $f(x_\beta)=0$, or $\beta \not \in \textmd{Im}(L_f)$ }, \\
p^{m-2}+\varepsilon_f \bar{\eta}(-f(x_\beta))(p-1)p^{m-2}(p^*)^{-\frac{r_f-1}{2}}   & \mbox{ if $\beta \in \textmd{Im}(L_f)$ and $f(x_\beta)\neq0$},
\end{array}
\right.
$$
where $x_\beta \in \gf(q)$ satisfies $L_f(x_\beta)=-\frac{\beta}{2}$ when $\beta \in \textmd{Im}(L_f)$.
\end{lemma}

\begin{proof}By definition, we have
\begin{eqnarray*}
N_{f,\beta}
&=& p^{-2} \sum_{x\in \gf(q)}(\sum_{y\in \gf(p)}\zeta_p^{yf(x)})(\sum_{z\in \gf(p)}\zeta_p^{-z\tr(\beta x)})   \\
&=& p^{-2}\left(\sum_{z\in \gf(p)}\sum_{x\in \gf(q)}\zeta_p^{-z\tr(\beta x)}+\sum_{y\in \gf(p)^*}\sigma_y\ (\sum_{z\in \gf(p)}\sum_{x\in \gf(q)}\zeta_p^{f(x)-z\tr(\beta x)})\right).
\end{eqnarray*}
The desired conclusion then follows from Lamma \ref{lem-ftr}.
\end{proof}

\begin{lemma}\label{lem-D}
Let $f$ be a homogeneous quadratic function with the rank $r_f$ and the sign $\varepsilon_f$, $\alpha \in \gf(q) \backslash Im(L_f)$ and $\beta \in \gf(q)^*$. Then we have the following.

\begin{itemize}
  \item There exists $z_0\in \gf(p)^*$ such that $\alpha-z_0\beta \in Im(L_f)$ if and only if $\beta \in \bigcup_{z\in \gf(p)^*}(z\alpha+\textmd{Im}(L_f))$.
  \item Let $z'\in \gf(p)^*$ and $\beta \in z'\alpha+Im(L_f)$. Then
  $
  \{z\in \gf(p)^*:\alpha - z \beta \in Im(L_f)\}=\{\frac{1}{z'}\}.
  $
\end{itemize}
\end{lemma}
\begin{proof}
The desired conclusion is straightforward.
\end{proof}

\begin{lemma}\label{lem-s4}
Let $f$ be a homogeneous quadratic function with the rank $r_f$ and the sign $\varepsilon_f$, $\alpha \in \gf(q)$, $\beta \in \gf(q)^*$ and
$$
S_4=\sum_{z\in \gf(p)}\sum_{x\in \gf(q)}\zeta_p^{f(x)-\tr((\alpha-\beta z) x)}.
$$
Then we have the following.

(\uppercase\expandafter{\romannumeral1}) If $\alpha \in Im(L_f)$, then
$$
S_4=
\left\{ \begin{array}{ll}
\varepsilon_f p^{m+1}(p^*)^{-\frac{r_f}{2}} \zeta_p^{-f(x_\alpha)}   & \mbox{ if $\beta \in \textmd{Im}(L_f)$, $f(x_\beta)=0$ and $\tr(\alpha x_\beta)=0$}, \\
0                                                & \mbox{ if $\beta \in \textmd{Im}(L_f)$, $f(x_\beta)=0$ and $\tr(\alpha x_\beta)\neq 0$}, \\
\varepsilon_f \bar{\eta}(-f(x_\beta))p^{m}(p^*)^{-\frac{r_f-1}{2}} \zeta_p^{-f(x_\alpha)+\frac{(\tr(\alpha x_\beta))^2}{4f(x_\beta)}}   & \mbox{ if $\beta \in \textmd{Im}(L_f)$ and $f(x_\beta) \neq 0$}, \\
\varepsilon_f p^{m}(p^*)^{-\frac{r_f}{2}} \zeta_p^{-f(x_\alpha)}           & \mbox{ if $\beta \not \in \textmd{Im}(L_f)$,}
\end{array}
\right.
$$
where $x_\alpha \in \gf(q)$ satisfies $L_f(x_\alpha)=-\frac{\alpha}{2}$ and $x_\beta \in \gf(q)$ satisfies $L_f(x_\beta)=-\frac{\beta}{2}$.

(\uppercase\expandafter{\romannumeral2}) If $\alpha \not \in Im(L_f)$, then
$$
S_4=
\left\{ \begin{array}{ll}
\varepsilon_f p^{m}(p^*)^{-\frac{r_f}{2}} \zeta_p^{-f(x')}   & \mbox{ if $\beta \in \bigcup_{z\in \gf(p)^*}(z\alpha+\textmd{Im}(L_f))$}, \\
0                                                & \mbox{ otherwise },
\end{array}
\right.
$$
where $f(x')=-\frac{\alpha -\beta z_0}{2}$ with $\beta \in \frac{1}{z_0}\alpha +Im(L_f)$ and $z_0\in \gf(p)^*$.
\end{lemma}

\begin{proof}
(\uppercase\expandafter{\romannumeral1}) It is obvious that there exists $x_\alpha \in \gf(q)$ such that $L_f(x_\alpha)=-\frac{\alpha}{2}$ when $\alpha \in Im(L_f)$. Let us distinguish the following two cases when $\alpha \in Im(L_f)$.
\begin{itemize}
  \item Case $\beta \in Im(L_f)$.

  It is obvious that there exists $x_\beta \in \gf(q)$ such that $L_f(x_\beta)=-\frac{\beta}{2}$. Thus, $L_f(x_\alpha -zx_\beta)=-\frac{\alpha-z\beta}{2}$. From Lemma \ref{lem-esf}, we have
  \begin{eqnarray}\label{eqn-s41}
  S_4
  &=& \sum_{z\in \gf(p)}\varepsilon_f p^{m}(p^*)^{-\frac{r_f}{2}} \zeta_p^{-f(x_\alpha-zx_\beta)}   \nonumber \\
  &=& \varepsilon_f p^{m}(p^*)^{-\frac{r_f}{2}} \sum_{z\in \gf(p)} \zeta_p^{-f(x_\alpha)-f(x_\beta)z^2+2\tr(L_f(x_\alpha)x_\beta)z}  \nonumber \\
  &=& \varepsilon_f p^{m}(p^*)^{-\frac{r_f}{2}} \sum_{z\in \gf(p)} \zeta_p^{-f(x_\beta)z^2-\tr(\alpha x_\beta)z-f(x_\alpha)}    \nonumber \\
  &=&
  \left\{ \begin{array}{ll}
\varepsilon_f p^{m+1}(p^*)^{-\frac{r_f}{2}} \zeta_p^{-f(x_\alpha)}   & \mbox{ if $f(x_\beta)=0$ and $\tr(\alpha x_\beta)=0$}, \\
0                                                & \mbox{ if $f(x_\beta)=0$ and $\tr(\alpha x_\beta)\neq 0$}, \\
\varepsilon_f \bar{\eta}(-f(x_\beta))p^{m}(p^*)^{-\frac{r_f-1}{2}} \zeta_p^{-f(x_\alpha)+\frac{(\tr(\alpha x_\beta))^2}{4f(x_\beta)}}   & \mbox{ if $f(x_\beta) \neq 0$},
\end{array}
\right.
\end{eqnarray}
  where the last identity follows by using Lemmas \ref{lem-32A1} and \ref{lem-32A2}.
  \item Case $\beta \not \in Im(L_f)$.

  It is clear that $\alpha-\beta z \not \in Im(L_f)$ for any $z\in \gf(p)^*$. Therefore, from Lemma \ref{lem-esf} we have
  \begin{eqnarray}\label{eqn-s42}
  S_4
  &=& \sum_{x\in \gf(q)}\zeta_p^{f(x)-\tr(\alpha x))}      \nonumber \\
  &=& \varepsilon_f p^{m}(p^*)^{-\frac{r_f}{2}} \zeta_p^{-f(x_\alpha)}.
  \end{eqnarray}
\end{itemize}
Combining (\ref{eqn-s41}) and (\ref{eqn-s42}), the result (\uppercase\expandafter{\romannumeral1}) of this lemma follows.

(\uppercase\expandafter{\romannumeral2}) The proof is similar to case (\uppercase\expandafter{\romannumeral1}). The desired conclusion then follows from Lammas \ref{lem-esf} and \ref{lem-D}.
\end{proof}

\begin{lemma}\label{lem-trs4}
Let the symbols and notations be as Lemma \ref{lem-s4}, and let
$$
S_5=\sum_{y\in \gf(p)^*}\sigma_y\left(\sum_{z\in \gf(p)}\sum_{x\in \gf(q)}\zeta_p^{f(x)-\tr((\alpha-\beta z) x)}\right).
$$
Then we have the following.

(\uppercase\expandafter{\romannumeral1}) When $\alpha \in Im(L_f)$, we have the following four cases.
\begin{itemize}
  \item If $r_f$ is even and $f(x_\alpha)=0$, then

 \begin{eqnarray*}
  S_5=
\left\{ \begin{array}{ll}
\varepsilon_f (p-1)p^{m+1}(p^*)^{-\frac{r_f}{2}}  & \mbox{ if $f(x_\beta)=0$ and $\tr(\alpha x_\beta)=0$}, \\
0                                                 & \mbox{ if $f(x_\beta)=0$ and $\tr(\alpha x_\beta)\neq 0$} \\
                                                  & \mbox{ or $f(x_\beta)\neq 0$ and $\tr(\alpha x_\beta)= 0$}, \\
\varepsilon_f \bar{\eta}(-1)p^{m}(p^*)^{-\frac{r_f-2}{2}} & \mbox{ if $f(x_\beta) \neq 0$ and $\tr(\alpha x_\beta)\neq 0$},\\
\varepsilon_f (p-1)p^{m}(p^*)^{-\frac{r_f}{2}} & \mbox{ if $\beta \not \in Im(L_f)$}.
\end{array}
\right.
\end{eqnarray*}

  \item If $r_f$ is even and $f(x_\alpha) \neq 0$, then
  \begin{eqnarray*}
S_5=
\left\{ \begin{array}{ll}
-\varepsilon_f p^{m+1}(p^*)^{-\frac{r_f}{2}}  & \mbox{ if $f(x_\beta)=0$ and $\tr(\alpha x_\beta)=0$}, \\
0                                                 & \mbox{ if $f(x_\beta)=0$ and $\tr(\alpha x_\beta)\neq 0$} \\
                                                  & \mbox{ or $f(x_\beta)\neq 0$ and $E= 0$}, \\
\varepsilon_f \bar{\eta}(-f(x_\beta)E)p^{m}(p^*)^{-\frac{r_f-2}{2}} & \mbox{ if $f(x_\beta) \neq 0$ and $E \neq 0$},\\
-\varepsilon_f p^{m}(p^*)^{-\frac{r_f}{2}} & \mbox{ if $\beta \not \in Im(L_f)$},
\end{array}
\right.
\end{eqnarray*}
where $E=-f(x_\alpha)+\frac{(\tr(\alpha x_\beta))^2}{4f(x_\beta)}$.
  \item If $r_f$ is odd and $f(x_\alpha)=0$, then
 \begin{eqnarray*}
S_5=
\left\{ \begin{array}{ll}
0                                                 & \mbox{ if $f(x_\beta)=0$ or $\beta \not \in Im(L_f)$}, \\
\varepsilon_f \bar{\eta}(-f(x_\beta))(p-1)p^{m}(p^*)^{-\frac{r_f-1}{2}} & \mbox{ if $f(x_\beta) \neq 0$ and $\tr(\alpha x_\beta)= 0$},\\
-\varepsilon_f \bar{\eta}(-f(x_\beta))p^{m}(p^*)^{-\frac{r_f-1}{2}} & \mbox{ if $f(x_\beta) \neq 0$ and $\tr(\alpha x_\beta)\neq 0$}.
\end{array}
\right.
\end{eqnarray*}
  \item If $r_f$ is odd and $f(x_\alpha) \neq 0$, then
\begin{eqnarray*}
S_5=
\left\{ \begin{array}{ll}
\varepsilon_f \bar{\eta}(-f(x_\alpha))p^{m+1}(p^*)^{-\frac{r_f-1}{2}}  & \mbox{ if $f(x_\beta)=\tr(\alpha x_\beta)=0$}, \\
0                                                                     & \mbox{ if $f(x_\beta)=0$ and $\tr(\alpha x_\beta)\neq 0$} \\
\varepsilon_f \bar{\eta}(-f(x_\alpha))(p-1)p^{m}(p^*)^{-\frac{r_f-1}{2}}   & \mbox{ if $f(x_\beta)\neq 0$ and $E= 0$}, \\
-\varepsilon_f \bar{\eta}(-f(x_\beta))p^{m}(p^*)^{-\frac{r_f-1}{2}}        & \mbox{ if $f(x_\beta) \neq 0$ and $E \neq 0$},\\
\varepsilon_f \bar{\eta}(-f(x_\alpha))p^{m}(p^*)^{-\frac{r_f-1}{2}}  & \mbox{ if $\beta \not \in Im(L_f)$}.
\end{array}
\right.
\end{eqnarray*}
where $E=-f(x_\alpha)+\frac{(\tr(\alpha x_\beta))^2}{4f(x_\beta)}$.
\end{itemize}

(\uppercase\expandafter{\romannumeral2})  When $\alpha \not \in Im(L_f)$, we have the following two cases.
\begin{itemize}
  \item If $r_f$ is even, then
  $$
S_5=
\left\{ \begin{array}{ll}
-\varepsilon_f p^{m}(p^*)^{-\frac{r_f}{2}}    & \mbox{ if $\beta \in \bigcup_{z\in \gf(p)^*}(z\alpha+\textmd{Im}(L_f))$ and $f(x') \neq 0$}, \\
(p-1)\varepsilon_f p^{m}(p^*)^{-\frac{r_f}{2}}    & \mbox{ if $\beta \in \bigcup_{z\in \gf(p)^*}(z\alpha+\textmd{Im}(L_f))$ and $f(x') = 0$}, \\
0                                                & \mbox{ otherwise },
\end{array}
\right.
$$
where $f(x')=-\frac{\alpha -\beta z_0}{2}$ with $z_0\in \gf(p)^*$ and $\beta \in \frac{1}{z_0}\alpha +Im(L_f)$.
\item If $r_f$ is odd, then
  $$
S_5=
\left\{ \begin{array}{ll}
\varepsilon_f \bar{\eta}(-f(x'))p^{m}(p^*)^{-\frac{r_f-1}{2}} & \mbox{ if $\beta \in \bigcup_{z\in \gf(p)^*}(z\alpha+\textmd{Im}(L_f))$ and $f(x') \neq 0$}, \\
0                                                & \mbox{ otherwise },
\end{array}
\right.
$$
where $f(x')=-\frac{\alpha -\beta z_0}{2}$ with $\beta \in \frac{1}{z_0}\alpha +Im(L_f)$ and $z_0\in \gf(p)^*$.
\end{itemize}
\end{lemma}
\begin{proof}
The desired conclusions then follow from Lammas \ref{lem-s4} and \ref{lem-A2}.
\end{proof}

\begin{lemma}\label{lem-Nf3}
Let $f$ be a homogeneous quadratic function with the rank $r_f$ and the sign $\varepsilon_f$, $\alpha \in \gf(q)$, $\beta \in \gf(q)^*$ and
$$
N_{f,\beta}(\alpha)=\{x\in \gf(q):f(x)-\tr(\alpha x)=0 ~and~ \tr(\beta x)=0\}.
$$
Then we have the following.

(\uppercase\expandafter{\romannumeral1}) When $\alpha \in Im(L_f)$, we have the following four cases.
\begin{itemize}
  \item If $r_f$ is even and $f(x_\alpha)=0$, then

 \begin{eqnarray*}
N_{f,\beta}(\alpha)=
\left\{ \begin{array}{ll}
p^{m-2}+\varepsilon_f (p-1)p^{m-1}(p^*)^{-\frac{r_f}{2}}  & \mbox{ if $f(x_\beta)=0$ and $\tr(\alpha x_\beta)=0$}, \\
p^{m-2}                                                 & \mbox{ if $f(x_\beta)=0$ and $\tr(\alpha x_\beta)\neq 0$} \\
                                                  & \mbox{ or $f(x_\beta)\neq 0$ and $\tr(\alpha x_\beta)= 0$}, \\
p^{m-2}+\varepsilon_f \bar{\eta}(-1)p^{m-2}(p^*)^{-\frac{r_f-2}{2}} & \mbox{ if $f(x_\beta) \neq 0$ and $\tr(\alpha x_\beta)\neq 0$},\\
p^{m-2}+\varepsilon_f (p-1)p^{m-2}(p^*)^{-\frac{r_f}{2}} & \mbox{ if $\beta \not \in Im(L_f)$}.
\end{array}
\right.
\end{eqnarray*}

  \item If $r_f$ is even and $f(x_\alpha) \neq 0$, then
  \begin{eqnarray*}
N_{f,\beta}(\alpha)=
\left\{ \begin{array}{ll}
p^{m-2}-\varepsilon_f p^{m-1}(p^*)^{-\frac{r_f}{2}}  & \mbox{ if $f(x_\beta)=0$ and $\tr(\alpha x_\beta)=0$}, \\
p^{m-2}                                                 & \mbox{ if $f(x_\beta)=0$ and $\tr(\alpha x_\beta)\neq 0$} \\
                                                  & \mbox{ or $f(x_\beta)\neq 0$ and $E= 0$}, \\
p^{m-2}+\varepsilon_f \bar{\eta}(-f(x_\beta)E)p^{m-2}(p^*)^{-\frac{r_f-2}{2}} & \mbox{ if $f(x_\beta) \neq 0$ and $E \neq 0$},\\
p^{m-2}-\varepsilon_f p^{m-2}(p^*)^{-\frac{r_f}{2}} & \mbox{ if $\beta \not \in Im(L_f)$},
\end{array}
\right.
\end{eqnarray*}
where $E=-f(x_\alpha)+\frac{(\tr(\alpha x_\beta))^2}{4f(x_\beta)}$.
  \item If $r_f$ is odd and $f(x_\alpha)=0$, then
\begin{eqnarray*}
N_{f,\beta}(\alpha)=
\left\{ \begin{array}{ll}
p^{m-2}                                                 & \mbox{ if $f(x_\beta)=0$ or $\beta \not \in Im(L_f)$}, \\
p^{m-2}+\varepsilon_f \bar{\eta}(-f(x_\beta))(p-1)p^{m-2}(p^*)^{-\frac{r_f-1}{2}} & \mbox{ if $f(x_\beta) \neq 0$ and $\tr(\alpha x_\beta)= 0$},\\
p^{m-2}-\varepsilon_f \bar{\eta}(-f(x_\beta))p^{m-2}(p^*)^{-\frac{r_f-1}{2}} & \mbox{ if $f(x_\beta) \neq 0$ and $\tr(\alpha x_\beta)\neq 0$}.
\end{array}
\right.
\end{eqnarray*}
  \item If $r_f$ is odd and $f(x_\alpha) \neq 0$, then
\begin{eqnarray*}
N_{f,\beta}(\alpha)=
\left\{ \begin{array}{ll}
p^{m-2}+\varepsilon_f \bar{\eta}(-f(x_\alpha))p^{m-1}(p^*)^{-\frac{r_f-1}{2}}  & \mbox{ if $f(x_\beta)=\tr(\alpha x_\beta)=0$}, \\
p^{m-2}                                                                     & \mbox{ if $f(x_\beta)=0$ and $\tr(\alpha x_\beta)\neq 0$} \\
p^{m-2}+\varepsilon_f \bar{\eta}(-f(x_\alpha))(p-1)p^{m-2}(p^*)^{-\frac{r_f-1}{2}}   & \mbox{ if $f(x_\beta)\neq 0$ and $E= 0$}, \\
p^{m-2}-\varepsilon_f \bar{\eta}(-f(x_\beta))p^{m-2}(p^*)^{-\frac{r_f-1}{2}}        & \mbox{ if $f(x_\beta) \neq 0$ and $E \neq 0$},\\
p^{m-2}+\varepsilon_f \bar{\eta}(-f(x_\alpha))p^{m-2}(p^*)^{-\frac{r_f-1}{2}}  & \mbox{ if $\beta \not \in Im(L_f)$}.
\end{array}
\right.
\end{eqnarray*}
where $E=-f(x_\alpha)+\frac{(\tr(\alpha x_\beta))^2}{4f(x_\beta)}$.
\end{itemize}

(\uppercase\expandafter{\romannumeral2})  When $\alpha \not \in Im(L_f)$, we have the following two cases.
\begin{itemize}
  \item If $r_f$ is even, then
  $$
N_{f,\beta}(\alpha)=
\left\{ \begin{array}{ll}
p^{m-2}-\varepsilon_f p^{m-2}(p^*)^{-\frac{r_f}{2}}    & \mbox{ if $\beta \in \bigcup_{z\in \gf(p)^*}(z\alpha+\textmd{Im}(L_f))$ and $f(x') \neq 0$}, \\
p^{m-2}+(p-1)\varepsilon_f p^{m-2}(p^*)^{-\frac{r_f}{2}}    & \mbox{ if $\beta \in \bigcup_{z\in \gf(p)^*}(z\alpha+\textmd{Im}(L_f))$ and $f(x') = 0$}, \\
p^{m-2}                                                & \mbox{ otherwise },
\end{array}
\right.
$$
where $f(x')=-\frac{\alpha -\beta z_0}{2}$ with $\beta \in \frac{1}{z_0}\alpha +Im(L_f)$ and $z_0\in \gf(p)^*$.
\item If $r_f$ is odd, then
  $$
N_{f,\beta}(\alpha)=
\left\{ \begin{array}{ll}
p^{m-2}+\varepsilon_f \bar{\eta}(-f(x'))p^{m-2}(p^*)^{-\frac{r_f-1}{2}} & \mbox{ if $\beta \in \bigcup_{z\in \gf(p)^*}(z\alpha+\textmd{Im}(L_f))$ and $f(x') \neq 0$}, \\
p^{m-2}                                               & \mbox{ otherwise },
\end{array}
\right.
$$
where $f(x')=-\frac{\alpha -\beta z_0}{2}$ with $\beta \in \frac{1}{z_0}\alpha +Im(L_f)$ and $z_0\in \gf(p)^*$.
\end{itemize}
\end{lemma}
\begin{proof}By definition, we have
\begin{eqnarray*}
N_{f,\beta}(\alpha)
&=& p^{-2} \sum_{x\in \gf(q)}(\sum_{y\in \gf(p)}\zeta_p^{y(f(x)-\tr(\alpha x))})(\sum_{z\in \gf(p)}\zeta_p^{z\tr(\beta x)})   \\
&=& p^{-2}\left(\sum_{z\in \gf(p)}\sum_{x\in \gf(q)}\zeta_p^{z\tr(\beta x)}+\sum_{y\in \gf(p)^*}\sigma_y\ (\sum_{z\in \gf(p)}\sum_{x\in \gf(q)}\zeta_p^{f(x)-\tr((\alpha-\beta z)x)})\right).
\end{eqnarray*}
The desired conclusion then follows from Lamma \ref{lem-trs4} and the result (\uppercase\expandafter{\romannumeral1}) of Lemma \ref{lem-ftr}.

This completes the proof.
\end{proof}

\begin{lemma}\label{lem-NE}
Let $f$ be a homogeneous quadratic function with the rank $r_f$ and the sign $\varepsilon_f$, $\alpha \in Im(L_f)$ and $x_\alpha \in \gf(q)$ with satisfying $L_f(x_\alpha)=-\frac{\alpha}{2}$.  Suppose that $f(x_\alpha)\neq 0$, we define
$$
S_6=\sum_{z\in \gf(p)}\sum_{w\in \gf(p)}\sum_{x\in \gf(q)}\zeta_p^{f(x)-\frac{1}{4f(x_\alpha)}z^2+w(z-\tr(\alpha x))}
$$
and
$$
N_{E}=\# \{x\in \gf(q):f(x)-\frac{1}{4f(x_\alpha)}(\tr(\alpha x))^2=0\}.
$$
Then we have the following:

(\uppercase\expandafter{\romannumeral1}) $S_6=\varepsilon_f \bar{\eta}(-f(x_\alpha))p^{m+1}(p^*)^{-\frac{r_f-1}{2}}$,

(\uppercase\expandafter{\romannumeral2})
$\sum_{y\in \gf(p)^*}\sigma_y(S_6)
=\left\{ \begin{array}{ll}
0                                              & \mbox{ if $r_f$ is even,} \\
\varepsilon_f \bar{\eta}(-f(x_\alpha))(p-1)p^{m+1}(p^*)^{-\frac{r_f-1}{2}}   & \mbox{ if $r_f$ is odd,}
\end{array}
\right.$

(\uppercase\expandafter{\romannumeral3})
$N_{E}
=\left\{ \begin{array}{ll}
p^{m-1}                                              & \mbox{ if $r_f$ is even,} \\
p^{m-1} + \varepsilon_f \bar{\eta}(-f(x_\alpha))(p-1)p^{m-1}(p^*)^{-\frac{r_f-1}{2}}   & \mbox{ if $r_f$ is odd.}
\end{array}
\right.$
\end{lemma}

\begin{proof} (\uppercase\expandafter{\romannumeral1}) By definition, we have
\begin{eqnarray*}
S_6
&=& \sum_{z\in \gf(p)}\sum_{w\in \gf(p)}\zeta_p^{-\frac{1}{4f(x_\alpha)}z^2+wz}
\sum_{x\in \gf(q)}\zeta_p^{f(x)-\tr(w \alpha x)}\\
&=& \varepsilon_f p^{m}(p^*)^{-\frac{r_f}{2}}\sum_{z\in \gf(p)}\sum_{w\in \gf(p)}\zeta_p^{-\frac{1}{4f(x_\alpha)}z^2+wz-f(x_\alpha)w^2}
~~~~~~(\mbox{By the result (\uppercase\expandafter{\romannumeral2}) of Lemma \ref{lem-esf}})\\
&=& \varepsilon_f \bar{\eta}(-f(x_\alpha))p^{m+1}(p^*)^{-\frac{r_f-1}{2}}.~~~~~~~~~~~~~~~~~~~~~~~~~~~~~(\mbox{By Lemma \ref{lem-szw}})
\end{eqnarray*}

(\uppercase\expandafter{\romannumeral2}) The desired conclusion then follows from Lemma \ref{lem-A2} and the result (\uppercase\expandafter{\romannumeral1}) of this Lemma.

(\uppercase\expandafter{\romannumeral3}) For any $x\in \gf(q)$, we have
\begin{eqnarray*}
& & p^{-2}\sum_{z\in \gf(p)}\left(\sum_{w\in \gf(p)}\zeta_p^{w(z-\tr(\alpha x))}\right)\left(\sum_{y\in \gf(p)}\zeta_p^{y(f(x)-\frac{1}{4f(x_\alpha)}z^2)}\right)\\
& & = \left\{ \begin{array}{ll}
1                                              & ~~~~~~~ ~~~~~~~ ~~~~~~~\mbox{ if $f(x)-\frac{1}{4f(x_\alpha)}(\tr(\alpha x))^2=0$,} \\
0                                              & ~~~~~~~ ~~~~~~~ ~~~~~~~\mbox{ otherwise.}
\end{array}
\right.
\end{eqnarray*}
Therefore,
\begin{eqnarray*}
 N_{E} &=& p^{-2}\sum_{x\in \gf(q)}\sum_{z\in \gf(p)}\left(\sum_{w\in \gf(p)}\zeta_p^{w(z-\tr(\alpha x))}\right)\left(\sum_{y\in \gf(p)}\zeta_p^{y(f(x)-\frac{1}{4f(x_\alpha)}z^2)}\right) \\
&=& p^{-2}\sum_{y\in \gf(p)}\sum_{z\in \gf(p)}\sum_{w\in \gf(p)}\sum_{x\in \gf(q)} \zeta_p^{y(f(x)-\frac{1}{4f(x_\alpha)}z^2)+w(z-\tr(\alpha x))} \\
&=& p^{-2}\sum_{z\in \gf(p)}\sum_{w\in \gf(p)}\sum_{x\in \gf(q)} \zeta_p^{w(z-\tr(\alpha x))} \\
& & + p^{-2}\sum_{y\in \gf(p)^*}\sigma_y\left(\sum_{z\in \gf(p)}\sum_{w\in \gf(p)}\sum_{x\in \gf(q)} \zeta_p^{f(x)-\frac{1}{4f(x_\alpha)}z^2+w(z-\tr(\alpha x))}\right).
\end{eqnarray*}
Note that
$$
\sum_{z\in \gf(p)}\sum_{w\in \gf(p)}\sum_{x\in \gf(q)} \zeta_p^{w(z-\tr(\alpha x))}=p^{m+1}.
$$
The desired conclusion then follows from the result (\uppercase\expandafter{\romannumeral2}) of this lemma.

This completes the proof.
\end{proof}

\begin{lemma}\label{lem-esq}
Let $f$ be a homogeneous quadratic function with the rank $r_f$ and the sign $\varepsilon_f$, $\alpha \in Im(L_f)$ and $x_\alpha \in \gf(q)$ with satisfying $L_f(x_\alpha)=-\frac{\alpha}{2}$. Let $f(x_\alpha) \neq 0$,
$$
g(x)=f(x)-\frac{(\tr(\alpha x))^2}{4f(x_\alpha)}
$$
and
$
N(g=t)=\#\{x\in \gf(q):g(x)=t \}
$
for any $t\in \gf(p)$. Then we have the following results.

(\uppercase\expandafter{\romannumeral1}) $\sum_{x\in \gf(q)}\zeta_p^{g(x)}=\varepsilon_f \bar{\eta}(-f(x_\alpha))p^m(p^*)^{-\frac{r_f-1}{2}}$.

(\uppercase\expandafter{\romannumeral2})
$
N(g=t)=
\left\{ \begin{array}{ll}
p^{m-1}                                                             & \mbox{ if $r_f$ is even and $t=0$},\\
p^{m-1}+\varepsilon_f \bar{\eta}(-t) \bar{\eta}(-f(x_\alpha)) p^{m-1}(p^*)^{-\frac{r_f-2}{2}}     & \mbox{ if $r_f$ is even and $t\neq 0$},\\
p^{m-1}+\varepsilon_f \bar{\eta}(-f(x_\alpha))(p-1) p^{m-1}(p^*)^{-\frac{r_f-1}{2}}      & \mbox{ if $r_f$ is odd and $t=0$}\\
p^{m-1}-\varepsilon_f \bar{\eta}(-f(x_\alpha))p^{m-1}(p^*)^{-\frac{r_f-1}{2}}      & \mbox{ if $r_f$ is odd and $t\neq 0$}.
\end{array}
\right.
$
\end{lemma}

\begin{proof}
(\uppercase\expandafter{\romannumeral1}) By definition, we have
\begin{eqnarray*}
\sum_{x\in \gf(q)}\zeta_p^{g(x)}
&=& \sum_{x\in \gf(q)}\zeta_p^{f(x)-\frac{(\tr(\alpha x))^2}{4f(x_\alpha)}}\\
&=& \sum_{z\in \gf(p)}\left(\sum_{x\in \gf(q),\tr(x)=z}\zeta_p^{f(x)-\frac{z^2}{4f(x_\alpha)}}\right) \\
&=& \sum_{z\in \gf(p)}\left(\sum_{x\in \gf(q),\tr(x)=z}\zeta_p^{f(x)-\frac{z^2}{4f(x_\alpha)}}(p^{-1}\sum_{w\in \gf(p)}\zeta_p^{w(z-\tr(\alpha x))})\right) \\
&=& p^{-1}\sum_{z\in \gf(p)}\sum_{w\in \gf(p)}\sum_{x\in \gf(q)}\zeta_p^{f(x)-\frac{z^2}{4f(x_\alpha)}+w(z-\tr(\alpha x))} \\
&=& \varepsilon_f \bar{\eta}(-f(x_\alpha))p^m(p^*)^{-\frac{r_f-1}{2}},
\end{eqnarray*}
where the last identity follows from the result (\uppercase\expandafter{\romannumeral1}) of Lemma \ref{lem-NE}.

(\uppercase\expandafter{\romannumeral2}) By the result (\uppercase\expandafter{\romannumeral1}) of this lemma, it is clear that the rank of $g(x)$ is $r_g=r_f-1$ and the sign of $g(x)$ is $\varepsilon_g=\varepsilon_f \bar{\eta}(-f(x_\alpha))$. Thus the desired conclusion (\uppercase\expandafter{\romannumeral2}) then follows from Lemmas \ref{lem-gt} and \ref{lem-NE}.

This completes the proof.
\end{proof}

\begin{lemma}\label{lem-ea}
Let $f$ be a homogeneous quadratic function with the rank $r_f$ and the sign $\varepsilon_f$, $\alpha \in Im(L_f)$ and $x_\alpha \in \gf(q)$ with satisfying $L_f(x_\alpha)=-\frac{\alpha}{2}$. Let $f(x_\alpha) \neq 0$,
$$
g(x)=f(x)-\frac{(\tr(\alpha x))^2}{4f(x_\alpha)}
$$
and
$$
E=-f(x_\alpha)-\frac{(\tr(\alpha x))^2}{4f(x)}.
$$
When $r_f$ is even, we define
\begin{eqnarray*}
& & I_1=\#\{x\in \gf(q):f(x)=\tr(\alpha x)=0 \}\\
& & I_2=\#\left\{\{x\in \gf(q):f(x)=0 ~and~ \tr(\alpha x)\neq 0 \} \bigcup \{x\in \gf(q):f(x)\neq 0 ~and~ E=0 \}\right\} \\
& &I_3=\#\{x\in \gf(q):f(x)\neq 0, E\neq 0 ~and~ f(x)\cdot E \in \rm{NSQ} \}\\
& &I_4=\#\{x\in \gf(q):f(x)\neq 0, E\neq 0 ~and~ f(x)\cdot E \in \rm{SQ} \}.
\end{eqnarray*}
When $r_f$ is odd, we define
\begin{eqnarray*}
& & J_1=\#\{x\in \gf(q):f(x)\neq 0,\bar{\eta}(f(x))=\bar{\eta}(f(x_\alpha)) ~and~ E=0 \}\\
& & J_2=\#\{x\in \gf(q):f(x)\neq 0~and~\bar{\eta}(f(x))=\bar{\eta}(f(x_\alpha))\}\\
& & J_3=\#\{x\in \gf(q):f(x)=\tr(\alpha x)=0 \}\\
& & J_4=\#\{x\in \gf(q):f(x)=0 ~and~\tr(\alpha x)\neq0 \}\\
& & J_5=\#\{x\in \gf(q):f(x)\neq 0 ~and~ E=0 \}\\
& & J_6=\#\{x\in \gf(q):f(x)\neq 0, E\neq0 ~and~\bar{\eta}(f(x))=-\bar{\eta}(f(x_\alpha))\}.
\end{eqnarray*}
Then we have the following results.

(\uppercase\expandafter{\romannumeral1}) If $r_f$ is even, then
\begin{eqnarray}
& & I_1=p^{m-2},\label{eqn-a1}\\
& & I_2=(p-1)p^{m-2}(2+\varepsilon_f \cdot p(p^*)^{-\frac{r_f}{2}}), \label{eqn-a2} \\
& & I_3=
\frac{p-1}{2}p^{m-1}(1-\varepsilon_f \cdot p(p^*)^{-\frac{r_f}{2}})           \label{eqn-a3} \\
& & I_4=
\frac{(p-1)(p-2)}{2}p^{m-2}(1+\varepsilon_f \cdot p(p^*)^{-\frac{r_f}{2}}). \label{eqn-a4}
\end{eqnarray}

(\uppercase\expandafter{\romannumeral2})
If $r_f$ is odd, then
\begin{eqnarray}
& & J_1=(p-1)p^{m-2}(1+\varepsilon_f \bar{\eta}(-f(x_\alpha))(p-1)(p^*)^{-\frac{r_f-1}{2}}), \label{eqn-aa1} \\
& & J_2=\frac{(p-1)(p-2)}{2}p^{m-2}(1-\varepsilon_f \bar{\eta}(-f(x_\alpha))(p^*)^{-\frac{r_f-1}{2}}), \label{eqn-aa2} \\
& & J_3= p^{m-2}+\varepsilon_f \bar{\eta}(-f(x_\alpha))(p-1) p^{m-2}(p^*)^{-\frac{r_f-1}{2}}, \label{eqn-aa3} \\
& & J_4=(p-1)p^{m-2}(1-\varepsilon_f \bar{\eta}(-f(x_\alpha))(p^*)^{-\frac{r_f-1}{2}}), \label{eqn-aa4} \\
& & J_5= (p-1)p^{m-2}(1+\varepsilon_f \bar{\eta}(-f(x_\alpha))(p-1)(p^*)^{-\frac{r_f-1}{2}}),  \label{eqn-aa5} \\
& & J_6=\frac{p-1}{2}p^{m-1}(1-\varepsilon_f \bar{\eta}(-f(x_\alpha))(p^*)^{-\frac{r_f-1}{2}}).\label{eqn-aa6}
\end{eqnarray}
\end{lemma}

\begin{proof}
(\uppercase\expandafter{\romannumeral1}) If $r_f$ is even, then we have the following.
\begin{itemize}
  \item It is clear that Equation (\ref{eqn-a1}) follows directly from Lemma \ref{lem-Nf2}.
  \item By definition, we have
\begin{eqnarray*}
I_2 &=& \# \{x\in \gf(q):f(x)\neq 0 ~and~ g(x)=0 \} + \# \{x\in \gf(q):f(x)=0 ~and~ \tr(\alpha x)\neq 0 \} \\
&=&  \# \{x\in \gf(q):g(x)=0 \}-\# \{x\in \gf(q):f(x)=0 ~and~ \tr(\alpha x)=0 \}\\
& & + \# \{x\in \gf(q):f(x)=0 ~and~ \tr(\alpha x)\neq 0 \} \\
&=&  \# \{x\in \gf(q):g(x)=0 \}+\# \{x\in \gf(q):f(x)=0 \}-2\# \{x\in \gf(q):f(x)=\tr(\alpha x)=0 \}.
\end{eqnarray*}
Then Equation (\ref{eqn-a2}) follows from Lemmas \ref{lem-esq}, \ref{lem-Nf1} and \ref{lem-Nf2}.
  \item In Equations (\ref{eqn-a3}) and (\ref{eqn-a4}), we only give the proof for the case $-f(x_\alpha)\in \rm{SQ}$ and omit the proof for the case $-f(x_\alpha)\in \rm{NSQ}$ whose proof is similar.
  Suppose that $-f(x_\alpha)\in \rm{SQ}$, by definition and
  $$
  -\frac{f(x)E}{f(x_\alpha)}=f(x)-\frac{(\tr(\alpha x))^2}{4f(x_\alpha)}=g(x)
  $$
  we get
  \begin{eqnarray*}
  I_3 &=& \#\{x\in \gf(q):f(x)\neq 0~and~ g(x)\in \rm{NSQ} \} \\
  &=& \#\{x\in \gf(q):g(x)\in \rm{NSQ} \}-\#\{x\in \gf(q):f(x)=0~and~ g(x)\in \rm{NSQ} \}  \\
  &=& \#\{x\in \gf(q):g(x)\in \rm{NSQ} \}-\#\{x\in \gf(q):f(x)=0~and~ (\tr(\alpha x))^2 \in \rm{NSQ} \} \\
  &=& \#\{x\in \gf(q):g(x)\in \rm{NSQ} \} \\
  &=& \frac{p-1}{2}p^{m-1}(1-\varepsilon_f \cdot p(p^*)^{-\frac{r_f}{2}}),
  \end{eqnarray*}
where the last equation follows from Lemma \ref{lem-esq}.
This means that the first equation of (\ref{eqn-a3}) follows. Similarly, when $-f(x_\alpha)\in \rm{SQ}$, Equation (\ref{eqn-a4}) follows from (\ref{eqn-a1}) and (\ref{eqn-a2}).
\end{itemize}

(\uppercase\expandafter{\romannumeral2}) If $r_f$ is odd, then we give the proofs of the desired conclusions as follows.
\begin{itemize}
  \item Since
  $$
  -\frac{E}{4f(x_\alpha)}=\frac{g(x)}{4f(x)},
  $$
  we have
  \begin{eqnarray*}
  J_1
  &=& \#\{x\in \gf(q):f(x)\neq 0,\bar{\eta}(f(x))=\bar{\eta}(f(x_\alpha)) ~and~ E=0 \} \\
  &=& \#\{x\in \gf(q):g(x)=0 ~and~f(x)\neq 0 \} \\
  &=& \#\{x\in \gf(q):g(x)=0 \}- \#\{x\in \gf(q):g(x)=f(x)=0 \}  \\
  &=&(p-1)p^{m-2}(1+\varepsilon_f \bar{\eta}(-f(x_\alpha))(p-1)(p^*)^{-\frac{r_f-1}{2}}),
  \end{eqnarray*}
  where the last equation follows from Lemmas \ref{lem-esq} and \ref{lem-Nf2}. This means that Equation (\ref{eqn-aa1}) follows.
  \item By definition, we have
  \begin{eqnarray*}
  J_2
  &=& \#\{x\in \gf(q):f(x)\neq 0 ~and~\bar{\eta}(f(x))=\bar{\eta}(f(x_\alpha)) \} \\
  & & -\#\{x\in \gf(q):f(x)\neq 0,\bar{\eta}(f(x))=\bar{\eta}(f(x_\alpha)) ~and~ E=0 \} \\
  &=& \#\{x\in \gf(q):f(x)\neq 0 ~and~\bar{\eta}(f(x))=\bar{\eta}(f(x_\alpha)) \}-\#\{x\in \gf(q):E=0 ~and~f(x)\neq 0 \}.
  \end{eqnarray*}
  Then Equation (\ref{eqn-aa2}) follows from Lemma \ref{lem-gt} and (\ref{eqn-aa1}).
  \item Equation (\ref{eqn-aa3}) follows directly from Lemma \ref{lem-Nf2}.
  \item By definition, we have
  $$
  J_4=\#\{x\in \gf(q):f(x)=0 \}- \#\{x\in \gf(q):f(x)=\tr(\alpha x)=0 \}.
  $$
  The desired conclusion in (\ref{eqn-aa4}) then follows from Lemma \ref{lem-Nf1} and Equation (\ref{eqn-aa3}).
  \item Note that
  $$
  -\frac{E}{4f(x_\alpha)}=\frac{g(x)}{4f(x)}.
  $$
  Therefore, we have
  \begin{eqnarray*}
  J_5 &=& \#\{x\in \gf(q):f(x)\neq 0 ~and~ g(x)=0 \}\\
  &=& \#\{x\in \gf(q):g(x)=0 \}-\#\{x\in \gf(q):f(x)=\tr(\alpha x)=0 \}.
  \end{eqnarray*}
  The desired conclusion in (\ref{eqn-aa5}) then follows from Lemma \ref{lem-esq} and Equation (\ref{eqn-aa3}).
  \item The desired conclusion in (\ref{eqn-aa6}) then follows directly from (\ref{eqn-aa2}), (\ref{eqn-aa3}), (\ref{eqn-aa4}) and (\ref{eqn-aa5}).
\end{itemize}
This completes the proof of this lemma.
\end{proof}

\begin{lemma}\label{lem-fsqto}
Let $f$ be a homogeneous quadratic function with the rank $r_f$ and the sign $\varepsilon_f$, $\alpha \in Im(L_f)$ and $x_\alpha \in \gf(q)$ with satisfying $L_f(x_\alpha)=-\frac{\alpha}{2}$ and $f(x_\alpha)=0$. Then
\begin{itemize}
  \item $\#\{x\in \gf(q):f(x)\neq 0, \tr(\alpha x)=0 ~and~ -f(x)\in \rm{SQ} \}=\frac{p-1}{2}p^{m-2}(1+\varepsilon_f \cdot p(p^*)^{-\frac{r_f-1}{2}})$,
  \item $\#\{x\in \gf(q):f(x)\neq 0, \tr(\alpha x)=0 ~and~ -f(x)\in \rm{NSQ} \}=\frac{p-1}{2}p^{m-2}(1-\varepsilon_f \cdot p(p^*)^{-\frac{r_f-1}{2}})$,
  \item $\#\{x\in \gf(q):f(x)\tr(\alpha x)\neq 0 ~and~ -f(x)\in \rm{SQ} \}=\frac{(p-1)^2}{2}p^{m-2}$,
  \item $\#\{x\in \gf(q):f(x)\tr(\alpha x)\neq 0 ~and~ -f(x)\in \rm{NSQ} \}=\frac{(p-1)^2}{2}p^{m-2}$.
\end{itemize}
\end{lemma}
\begin{proof}
The desired conclusions then follow from Lemma \ref{lem-fato}.
\end{proof}

\subsection{Main results and their proofs}

The following two theorems are the main results of this paper.
\begin{theorem}\label{thm-main1}
Let $f$ be a homogeneous quadratic function with the rank $r_f$ and the sign $\varepsilon_f$, $\alpha \in Im(L_f)$ and $x_\alpha \in \gf(q)$ with satisfying $L_f(x_\alpha)=-\frac{\alpha}{2}$. Let $D$ be defined in (\ref{eqn-defsetD}). Then the set $\C_D$ of (\ref{eqn-maincode}) is a $[n,m]$ linear code over $\gf(p)$ with the weight distribution in Tables \ref{tab-even1}, \ref{tab-even0}, \ref{tab-odd1} and \ref{tab-odd0}, where
\begin{eqnarray}\label{eqn-codelenth1}
n=
\left\{ \begin{array}{ll}
p^{m-1}(1-\varepsilon_f(p^*)^{-\frac{r_f}{2}})-1               & \mbox{ if $r_f$ is even and $f(x_\alpha) \neq 0$,} \\
p^{m-1}(1+\varepsilon_f(p-1)(p^*)^{-\frac{r_f}{2}})-1                          & \mbox{ if $r_f$ is even and $f(x_\alpha)=0$,} \\
p^{m-1}(1+\varepsilon_f \bar{\eta}(-f(x_\alpha))(p^*)^{-\frac{r_f-1}{2}})-1               & \mbox{ if $r_f$ is odd and $f(x_\alpha) \neq 0$,} \\
p^{m-1}-1                 & \mbox{ if $r_f$ is odd and $f(x_\alpha) = 0$.}
\end{array}
\right.
\end{eqnarray}

\end{theorem}

\begin{table}[ht]
\begin{center}
\caption{The weight distribution of $\C_D$ of Theorem \ref{thm-main1} when $r_f$ is even and $f(x_\alpha)\neq 0$}\label{tab-even1}
\begin{tabular}{|c|c|} \hline
Weight $w$ &  Multiplicity $A_w$  \\ \hline
$0$          &  $1$ \\ \hline
$(p-1)p^{m-2}$  & $p^{r_f-2}+\frac{p-1}{2}p^{r_f-1}(1-\varepsilon_f\cdot p (p^*)^{-\frac{r_f}{2}})-1$\\ \hline
$p^{m-2}(p-1-\varepsilon_f\cdot p (p^*)^{-\frac{r_f}{2}})$  & $(p-1)p^{r_f-2}(2+\varepsilon_f\cdot p (p^*)^{-\frac{r_f}{2}})$ \\ \hline
$p^{m-2}(p-1-2\varepsilon_f\cdot p (p^*)^{-\frac{r_f}{2}})$  & $\frac{(p-1)(p-2)}{2}p^{r_f-2}(1+\varepsilon_f\cdot p (p^*)^{-\frac{r_f}{2}})$ \\ \hline
$p^{m-2}(p-1)(1-\varepsilon_f(p^*)^{-\frac{r_f}{2}})$   & $p^{m}-p^{r_f}$\\ \hline
\end{tabular}
\end{center}
\end{table}

\begin{table}[ht]
\begin{center}
\caption{The weight distribution of $\C_D$ of Theorem \ref{thm-main1} when $r_f$ is even and $f(x_\alpha)=0$}\label{tab-even0}
\begin{tabular}{|c|c|} \hline
Weight $w$ &  Multiplicity $A_w$  \\ \hline
$0$          &  $1$ \\ \hline
$(p-1)p^{m-2}$  & $p^{r_f-2}(1+\varepsilon_f\cdot (p-1)p (p^*)^{-\frac{r_f}{2}})-1$\\ \hline
$(p-1)p^{m-2}(1+\varepsilon_f\cdot p (p^*)^{-\frac{r_f}{2}})$  & $(p-1)p^{r_f-2}(2-\varepsilon_f\cdot p (p^*)^{-\frac{r_f}{2}})$ \\ \hline
$p^{m-2}(p-1+\varepsilon_f\cdot (p-2)p (p^*)^{-\frac{r_f}{2}})$  & $(p-1)^2p^{r_f-2}$ \\ \hline
$p^{m-2}(p-1)(1+\varepsilon_f\cdot (p-1) (p^*)^{-\frac{r_f}{2}})$   & $p^{m}-p^{r_f}$\\ \hline
\end{tabular}
\end{center}
\end{table}

\begin{table}[ht]
\begin{center}
\caption{The weight distribution of $\C_D$ of Theorem \ref{thm-main1} when $r_f$ is odd and $f(x_\alpha)\neq 0$}\label{tab-odd1}
\begin{tabular}{|c|c|} \hline
Weight $w$ &  Multiplicity $A_w$  \\ \hline
$0$          &  $1$ \\ \hline
$(p-1)p^{m-2}$  & $p^{r_f-2}(1+\varepsilon_f\bar{\eta}(-f(x_\alpha))(p-1)(p^*)^{-\frac{r_f-1}{2}})-1$\\ \hline
$p^{m-2}(p-1+\varepsilon_f \bar{\eta}(-f(x_\alpha)) p (p^*)^{-\frac{r_f-1}{2}})$  & $(p-1)p^{r_f-2}(1-\varepsilon_f\bar{\eta}(-f(x_\alpha))(p^*)^{-\frac{r_f-1}{2}})$ \\ \hline
$p^{m-2}(p-1+\varepsilon_f \bar{\eta}(-f(x_\alpha))(p^*)^{-\frac{r_f-1}{2}})$  & $(p-1)p^{r_f-2}(1+\varepsilon_f\bar{\eta}(-f(x_\alpha))(p-1)(p^*)^{-\frac{r_f-1}{2}})$ \\ \hline
$p^{m-2}(p-1+\varepsilon_f \bar{\eta}(-f(x_\alpha)) (p+1) (p^*)^{-\frac{r_f-1}{2}})$  & $\frac{(p-1)(p-2)}{2}p^{r_f-2}(1-\varepsilon_f\bar{\eta}(-f(x_\alpha))(p^*)^{-\frac{r_f-1}{2}})$ \\ \hline
$p^{m-2}(p-1)(1+\varepsilon_f \bar{\eta}(-f(x_\alpha)) (p^*)^{-\frac{r_f-1}{2}})$  & $\frac{(p-1)}{2}p^{r_f-1}(1-\varepsilon_f\bar{\eta}(-f(x_\alpha))(p^*)^{-\frac{r_f-1}{2}})+p^m-p^{r_f}$ \\ \hline
\end{tabular}
\end{center}
\end{table}

\begin{table}[ht]
\begin{center}
\caption{The weight distribution of $\C_D$ of Theorem \ref{thm-main1} when $r_f$ is odd and $f(x_\alpha)= 0$}\label{tab-odd0}
\begin{tabular}{|c|c|} \hline
Weight $w$ &  Multiplicity $A_w$  \\ \hline
$0$          &  $1$ \\ \hline
$p^{m-2}(p-1-\varepsilon_f(p-1)(p^*)^{-\frac{r_f-1}{2}})$  & $\frac{p-1}{2}p^{r_f-2}(1+\varepsilon_f\cdot p (p^*)^{-\frac{r_f-1}{2}})$\\ \hline
$p^{m-2}(p-1+\varepsilon_f(p-1)(p^*)^{-\frac{r_f-1}{2}})$  & $\frac{p-1}{2}p^{r_f-2}(1-\varepsilon_f\cdot p (p^*)^{-\frac{r_f-1}{2}})$\\ \hline
$p^{m-2}(p-1+\varepsilon_f(p^*)^{-\frac{r_f-1}{2}})$  & $\frac{(p-1)^2}{2}p^{r_f-2}$\\ \hline
$p^{m-2}(p-1-\varepsilon_f(p^*)^{-\frac{r_f-1}{2}})$  & $\frac{(p-1)^2}{2}p^{r_f-2}$\\ \hline
$p^{m-2}(p-1)$  & $p^{r_f-1}+p^m-p^{r_f}-1$ \\ \hline
\end{tabular}
\end{center}
\end{table}

\begin{proof}
By definition, the code length of $\C_D$ is $n = |D|=N_f(\alpha)-1$, where $N_f(\alpha)$ was defined by Lemma \ref{lem-Nf1}. This means that Equation (\ref{eqn-codelenth1}) follows.

For each $\beta \in \gf(q)^*$, define
\begin{eqnarray}\label{eqn-mcodeword}
\bc_{\beta}=(\tr(\beta d_1), \,\tr(\beta d_2), \,\ldots, \,\tr(\beta d_n)),
\end{eqnarray}
where $d_1, d_2, \ldots, d_n$ are the elements of $D$.
Then the Hamming weight $\wt(\bc_\beta)$ of $\bc_\beta$ is
\begin{eqnarray}\label{eqn-wcb}
\wt(\bc_\beta)=N_f(\alpha)-N_{f,\beta}(\alpha),
\end{eqnarray}
where $N_f(\alpha)$ and $N_{f,\beta}(\alpha)$ were defined before. By lemmas \ref{lem-Nf1} and \ref{lem-Nf3}, we have $\wt(\bc_\beta)=N_f(\alpha)-N_{f,\beta}(\alpha)>0$ for each $\beta \in \gf(q)^*$. This means that the code $\C_D$ has $q$ distinct codewords. Hence, the dimension of the code $\C_D$ is $m$.

Next we shall prove the the multiplicities $A_{w_i}$ of
codewords with weight $w_i$ in $\C_D$. Let us give the proofs of four cases, respectively.

\begin{enumerate}
  \item The case that $r_f$ is even and $f(x_\alpha) \neq 0$.

   We only give the proof for the case $-f(x_\alpha) \in \rm{SQ}$ and omit the proof for the case $-f(x_\alpha) \in \rm{NSQ}$ whose proof is similar. Suppose that $-f(x_\alpha) \in \rm{SQ}$. For each $\beta \in \gf(q)^*$, then from Lemmas \ref{lem-Nf1} and \ref{lem-Nf3} we obtain
   the Hamming weight
\begin{eqnarray*}\label{eqn-wb}
\wt(\bc_\beta)
&=&N_f(\alpha)-N_{f,\beta}(\alpha) \nonumber \\
&=&\left\{ \begin{array}{ll}
B_1                  & \mbox{ if $f(x_\beta)=\tr(\alpha x_\beta)=0$ or $f(x_\beta) \cdot E\in \rm{NSQ}$}, \\
B_1 -Bp              & \mbox{ if $f(x_\beta)=0$ and $\tr(\alpha x_\beta) \neq 0$ or $f(x_\beta)\neq 0$ and $E=0$,} \\
B_1 -2Bp            & \mbox{ if $f(x_\beta) \cdot E\in \rm{SQ}$,}\\
B_1 -B(p-1)               & \mbox{ if $\beta \not\in Im(L_f)$,}
\end{array}
\right.
\end{eqnarray*}
where $B_1={p}^{m-2}(p-1)$ and $B=p^{m-2}\varepsilon_f(p^*)^{-\frac{r_f}{2}}$. Define
$$
w_1=B_1,w_2=B_1 -Bp, w_3=B_1 -2Bp, w_4=B_1 -2B(p-1).
$$
Let
\begin{eqnarray*}
M_1 &=&\#\{\beta \in \gf(q): f(x_\beta)=\tr(\alpha x_\beta)=0\}+\#\{\beta \in \gf(q): f(x_\beta) \cdot E\in \rm{NSQ} \}\\
\end{eqnarray*}
Since the rank of linear mapping $\gf(q) \rightarrow \gf(q)$ ($x_\beta \mapsto -2L_f(x_\beta)$) is $r_f$, the dimension of their kernel is $m-r_f$. Therefore,
\begin{eqnarray*}
M_1 &=& p^{r_f-m} \#\{x\in \gf(q): f(x)=\tr(\alpha x)=0\}\\
& & +p^{r_f-m} \#\{x \in \gf(q): f(x) \cdot E\in \rm{NSQ} \}\\
&=& p^{r_f-2}+\frac{p-1}{2}p^{r_f-1}(1-\varepsilon_f\cdot p (p^*)^{-\frac{r_f}{2}}).~~~~~~(\mbox{By Lemma \ref{lem-ea}})
\end{eqnarray*}
Note that $f(0)=\tr(\alpha \cdot 0)=0$. Then
\begin{eqnarray*}
A_{w_1}&=&\#\{\beta \in \gf(q): \wt(\bc_\beta)=(p-1)p^{m-2}\}\\
&=& M_1-1\\
&=& p^{r_f-2}+\frac{p-1}{2}p^{r_f-1}(1-\varepsilon_f\cdot p (p^*)^{-\frac{r_f}{2}})-1.
\end{eqnarray*}
Similarly, the values of $A_{w_2}$, $A_{w_3}$ and $A_{w_4}$ can be calculated. This completes the proof of the weight distribution of Table \ref{tab-even1}.
  \item The case that $r_f$ is even and $f(x_\alpha) = 0$.

  The proof is similar to case 1) and we omit it here.
  The desired conclusion then follows from Lemmas \ref{lem-Nf1} and  \ref{lem-Nf2}.
  \item The case that $r_f$ is odd and $f(x_\alpha) \neq 0$.

  The proof is similar to case 1) and we omit it here.
  The desired conclusion then follows from Lemmas \ref{lem-Nf1} and  \ref{lem-ea}.
  \item The case that $r_f$ is odd and $f(x_\alpha) = 0$.

  The proof is similar to case 1) and we omit it here.
  The desired conclusion then follows from Lemmas \ref{lem-Nf1} and  \ref{lem-fsqto}.
\end{enumerate}
\end{proof}

As special cases of Theorem \ref{thm-main1}, the following two corollaries are direct consequences of Theorem \ref{thm-main1}.

\begin{corollary}\label{cor-1}
Let $u\in \gf(q)^*$, $f(x)=\tr(ux^2)$ and $\alpha \in \gf(q)^*$. Then
\begin{itemize}
  \item $\alpha \in Im(L_{f})$,
  \item $\varepsilon_f=(-1)^{m-1}\eta(-u)$,
  \item $r_{f}=m$,
  \item $L_{f}(x)=ux$,
  \item $x_\alpha=-\frac{\alpha}{2u}$ and $f(x_\alpha)=\frac{1}{4}\tr(\frac{\alpha^2}{u})$.
\end{itemize}
Thus, by using this function $f$, we can construct linear code $\C_D$ with the parameter and weight distribution given by Theorem \ref{thm-main1}.
\end{corollary}

\begin{corollary}\label{cor-3}
Let $v\in \gf(q)^*$, $\tr(v^2)\neq 0$, $f(x)=\tr(x^2)-\frac{1}{\tr(v^2)}(\tr(vx))^2$, $\alpha \in \gf(q)^*$ and $\tr(v \alpha)= 0$. Then
\begin{itemize}
  \item $\alpha \in Im(L_{f})$,
  \item $\varepsilon_{f}=(-1)^{m-1}\eta(-1)\bar{\eta}(-\tr(v^2))$,
  \item $r_{f}=m-1$,
  \item $L_{f}(x)=x-\frac{v}{\tr(v^2)}\tr(vx)$.
\end{itemize}
Thus, we can construct linear code $\C_D$ with the parameter and weight distribution given by Theorem \ref{thm-main1}.
\end{corollary}

As special cases of Corollary \ref{cor-1}, we give the following four examples.
\begin{example}
Let $(u,p,m)=(1,3,4)$, $\alpha \in \gf(q)^*$ and $\tr(\alpha)\neq 0$. Then the code $\C_D$ has parameters $[29, 4, 18]$ and weight enumerator
$1+44z^{18}+30z^{21}+6z^{24}$, which is verified by the Magma program. 
\end{example}

\begin{example}
Let $(u,p,m)=(1,3,6)$ and $\alpha \in \gf(p)^*$. Then the code $\C_D$ has parameters $[260, 6, 162]$ and weight enumerator
$1+98z^{162}+324z^{171}+306z^{180}$, which is verified by the Magma program. 
\end{example}

\begin{example}
Let $(u,p,m)=(1,3,5)$ and $\alpha \in \gf(p)^*$. Then the code $\C_D$ has parameters $[71, 5, 42]$ and weight enumerator
$1+30z^{42}+60z^{45}+90z^{48}+42z^{51}+20z^{54}$, which is verified by the Magma program. 
\end{example}

\begin{example}
Let $(u,p,m)=(1,3,3)$ and $\alpha \in \gf(p)^*$. Then the code $\C_D$ has parameters $[8, 3, 4]$ and weight enumerator
$1+6z^{4}+6z^{5}+8z^{6}+6z^{7}$, which is verified by the Magma program. 
\end{example}

As special cases of Corollary \ref{cor-3}, we give the following four examples.

\begin{example}
Let $(v,p,m)=(1,3,5)$, $g$ be a generator of $\gf(q)^*$ with the minimal polynomial $x^5+2x+1$. Let $\alpha=g^2$. Then the code $\C_D$ has parameters $[89, 5,54 ]$ and weight enumerator
$1+44z^{54}+162z^{60}+30z^{63}+6z^{72}$, which is verified by the Magma program. 
\end{example}

\begin{example}
Let $(v,p,m)=(1,3,5)$, $g$ be a generator of $\gf(q)^*$ with the minimal polynomial $x^5+2x+1$. Let $\alpha=g^3$. Then the code $\C_D$ has parameters $[62, 5,62 ]$ and weight enumerator
$1+42z^{36}+162z^{42}+36z^{45}+2z^{54}$, which is verified by the Magma program. 
\end{example}

\begin{example}
Let $(v,p,m)=(1,3,4)$, $g$ be a generator of $\gf(q)^*$ with the minimal polynomial $x^4+2x^3+2$. Let $\alpha=g^5$. Then the code $\C_D$ has parameters $[17, 4, 6 ]$ and weight enumerator
$1+4z^{4}+8z^{9}+66z^{12}+2z^{15}$, which is verified by the Magma program. 
\end{example}

\begin{example}
Let $(v,p,m)=(1,3,4)$, $g$ be a generator of $\gf(q)^*$ with the minimal polynomial $x^4+2x^3+2$. Let$\alpha=g^{13}$. Then the code $\C_D$ has parameters $[26, 4, 12]$ and weight enumerator
$1+6z^{12}+6z^{15}+62z^{18}+6z^{21}$, which is verified by the Magma program. 
\end{example}

\begin{theorem}\label{thm-main2}
Let $f$ be a homogeneous quadratic function with the rank $r_f$ and the sign $\varepsilon_f$. let $\alpha \not\in Im(L_f)$ and $D$ be defined in (\ref{eqn-defsetD}). Then the set $\C_D$ of (\ref{eqn-maincode}) is a $[n,m]$ linear code over $\gf(p)$ with the weight distribution in Tables \ref{tab-even} and \ref{tab-odd}, where $n=p^{m-1}-1$.
\end{theorem}

\begin{table}[ht]
\begin{center}
\caption{The weight distribution of $\C_D$ of Theorem \ref{thm-main2} when $r_f$ is even}\label{tab-even}
\begin{tabular}{|c|c|} \hline
Weight $w$ &  Multiplicity $A_w$  \\ \hline
$0$          &  $1$ \\ \hline
$p^{m-2}(p-1)(1-\varepsilon_f(p^*)^{-\frac{r_f}{2}})$  & $(p-1)p^{r_f-1}(1+\varepsilon_f(p-1)(p^*)^{-\frac{r_f}{2}})$\\ \hline
$p^{m-2}(p-1)+\varepsilon_fp^{m-2}(p^*)^{-\frac{r_f}{2}})$  & $(p-1)^2p^{r_f-1}(1-\varepsilon_f(p^*)^{-\frac{r_f}{2}})$\\ \hline

$p^{m-2}(p-1)$  & $p^m-(p-1)p^{r_f}-1$ \\ \hline
\end{tabular}
\end{center}
\end{table}

\begin{table}[ht]
\begin{center}
\caption{The weight distribution of $\C_D$ of Theorem \ref{thm-main2} when $r_f$ is odd}\label{tab-odd}
\begin{tabular}{|c|c|} \hline
Weight $w$ &  Multiplicity $A_w$  \\ \hline
$0$          &  $1$ \\ \hline
$p^{m-2}(p-1-\varepsilon_f(p^*)^{-\frac{r_f-1}{2}})$  & $\frac{(p-1)^2}{2}p^{r_f-1}(1+\varepsilon_f(p^*)^{-\frac{r_f-1}{2}})$\\ \hline
$p^{m-2}(p-1+\varepsilon_f(p^*)^{-\frac{r_f-1}{2}})$  & $\frac{(p-1)^2}{2}p^{r_f-1}(1-\varepsilon_f(p^*)^{-\frac{r_f-1}{2}})$\\ \hline
$p^{m-2}(p-1)$  & $p^m-(p-1)^{2}p^{r_f-1}-1$ \\ \hline
\end{tabular}
\end{center}
\end{table}

\begin{proof} The proof is similar to case 1) of Theorem \ref{thm-main1} and we omit it here. We point out that:
\begin{itemize}
  \item when $r_f$ is even, the desired conclusion then follows from Lemma \ref{lem-trs4},
  \item when $r_f$ is odd, the desired conclusion then follows from Lemmas \ref{lem-trs4} and \ref{lem-gt}.
\end{itemize}
This completes the proof.
\end{proof}

As special cases of Theorem \ref{thm-main2}, the following corollary is a direct consequence of Theorem \ref{thm-main2}.
\begin{corollary}\label{cor-2}
Let $v\in \gf(q)^*$, $\tr(v^2)\neq 0$, $f(x)=\tr(x^2)-\frac{1}{\tr(v^2)}(\tr(vx))^2$, $\alpha \in \gf(q)^*$ and $\tr(v \alpha)\neq 0$. Then
\begin{itemize}
  \item $\alpha \not\in Im(L_{f})$,
  \item $\varepsilon_{f}=(-1)^{m-1}\eta(-1)\bar{\eta}(-\tr(v^2))$,
  \item $r_{f}=m-1$,
  \item $L_{f}(x)=x-\frac{v}{\tr(v^2)}\tr(vx)$.
\end{itemize}
Thus, we can construct linear code $\C_D$ with the parameter and weight distribution given by Theorem \ref{thm-main2}.
\end{corollary}

As special cases of Corollary \ref{cor-2}, we give the following two examples.

\begin{example}
Let $(v,p,m)=(1,3,5)$ and $\alpha \in \gf(p)^*$. Then the code $\C_D$ has parameters $[26, 5, 15]$ and weight enumerator
$1+24z^{15}+44z^{18}+12z^{21}$, which is verified by the Magma program. 
\end{example}

\begin{example}
Let $(v,p,m)=(1,3,4)$ and $\alpha \in \gf(p)^*$. Then the code $\C_D$ has parameters $[80, 4, 51]$ and weight enumerator
$1+120z^{51}+80z^{54}+42z^{60}$, which is verified by the Magma program. 
\end{example}

\section{Concluding remarks}\label{sec-concluding}
In this paper, inspired by the works of \cite{DingDing2} and \cite{ZLFH2015}, inhomogeneous quadratic functions were used to construct
linear codes with few nonzero weights over finite fields. It was shown that the presented linear
codes have at most five nonzero weights. The weight distributions of the codes were
also determined and some of constructed linear codes are optimal in the sense that their
parameters meet certain bound on linear codes. The work of this paper extended the main
results in \cite{DingDing2} and \cite{ZLFH2015}.



\begin{thebibliography}{99}



\bibitem{Ding15} C. Ding, ``Linear codes from some 2-designs,'' \emph{IEEE Trans. Inf. Theory,} vol. 61, no. 6, pp. 3265-3275, June 2015.


\bibitem{DingDing1} K. Ding and C. Ding, ``Binary linear codes with three weights,'' \emph{IEEE Communication Letters,} vol. 18, no. 11, pp. 1879-1882, Novermber 2014.

\bibitem{DingDing2} K. Ding and C. Ding, ``A class of two-weight and three-weight codes and their applications in secret sharing,'' \emph{IEEE Trans. Inf. Theory,} vol. 61, no. 11, pp. 5835-5842, Nov. 2015.


\bibitem{WDX2015}
Q. Wang, K. Ding, and R. Xue, ``Binary linear codes with two weights,'' \emph{IEEE Communications Letters,} vol. 19, no. 7, pp. 1097--1100, 2015.


\bibitem{IR1990} K. Ireland and M. Rosen, ``A Classical Introduction to Modern Number Theory,'' 2nd ed. New York: Springer-Verlag, 1990, vol.  84, Graduate Texts in Mathematics.








\bibitem{FL07} K. Feng and J. Luo, ``Value distribution of exponential sums from
perfect nonlinear functions and their applications," {\em IEEE Trans. Inform.
Theory}, vol. 53, no. 9, pp. 3035--3041, 2007.

\bibitem{LN} R. Lidl and H. Niederreiter, {\em Finite Fields,} Cambridge: Cambridge University Press,
1997.



\bibitem{Ding20152} C. Ding, ``A Construction of Binary Linear Codes from Boolean Functions,'' arXiv:1511.00321.
\bibitem{Mesnager2015}
S. Mesnager, ``Linear codes with few weights from weakly regular bent functions based on a generic construction,''IACR Cryptology ePrint Archive 2015: 1103.

\bibitem{TLQZH2015}
C. Tang, N. Li, Y. Qi, Z. Zhou and T. Helleseth, ``Linear codes with two or three weights from
weakly regular bent functions,'' arXiv:1507.06148v3.
\bibitem{LWL2015}F. Li, Q. Wang and D. Lin, ``A class of three-weight and five-weight linear codes,'' arXiv:1509.06242v1.

\bibitem{TQH2015} C. Tang, Y. Qi,  D. Huang, ``Two-weight and three-weight linear codes from square functions,'' to appear IEEE Communications Letters, 2015.
\bibitem{QTH2015}Y. Qi, C. Tang and D. Huang, `` Binary linear codes with few weights,'' to appear IEEE Communications Letters, 2015.
\bibitem{XTF2015}C. Xiang, C. Tang and K. Feng, ``A class of linear codes with a few weights,'' arXiv:1512.07103v1.
\bibitem{ZLFH2015}
Z. Zhou, N. Li, C. Fan and T. Helleseth, ``Linear codes with two or three weights from quadratic bent functions,'' DOI 10.1007/s10623-015-0144-9.

\end{thebibliography}
\end{document}